\documentclass[letterpaper,twocolumn,10pt]{article}
\usepackage{usenix2019_v3}
\newcommand{\tronly}[2]{#1}
\usepackage{newtxtext}
\usepackage{appendix}
\tronly{}{%
\usepackage[subtle]{savetrees}
}
\usepackage{comment}
\usepackage{mathtools}
\usepackage{fancyvrb}

\DeclarePairedDelimiter{\floor}{\lfloor}{\rfloor}

\usepackage{cite}
\usepackage{paralist}
\usepackage[compact]{titlesec}
\usepackage{enumitem}
\setdefaultleftmargin{0em}{}{}{}{}{}
\setdefaultitem{{}}{}{}{}
\setlist[itemize]{leftmargin=*}
\setlist[enumerate]{leftmargin=*}
\usepackage{caption}
\usepackage[normalem]{ulem}
\usepackage{amsmath}
\usepackage{bbm}
\usepackage{thmtools, thm-restate}
\usepackage{mathrsfs}
\usepackage{amssymb}
\usepackage{tikz}
\usepackage{dsfont}
\usepackage{capt-of}
\usepackage{algorithmicx}
\usepackage{algorithm}
\usepackage[noend]{algpseudocode}
\usepackage{amsthm}
\usepackage[makeroom]{cancel}
\usepackage{datetime}

\graphicspath{{./figures}}

\newcommand{\Oh}[1]{O(#1)}

\let\emptyset\varnothing
\algnewcommand\algorithmicswitch{\textbf{switch}}
\algnewcommand\algorithmiccase{\textbf{case}}
\algnewcommand\algorithmicassert{\texttt{assert}}
\algnewcommand\Assert[1]{\State \algorithmicassert(#1)}%
\algdef{SE}[SWITCH]{Switch}{EndSwitch}[1]{\algorithmicswitch\ #1\ \algorithmicdo}{\algorithmicend\ \algorithmicswitch}%
\algdef{SE}[CASE]{Case}{EndCase}[1]{\algorithmiccase\ #1}{\algorithmicend\ \algorithmiccase}%
\algtext*{EndSwitch}%
\algtext*{EndCase}
\algnewcommand\algorithmiccontinue{\textbf{continue}}
\algnewcommand\algorithmicbreak{\textbf{break}}
\algnewcommand\Continue{\algorithmiccontinue}
\algnewcommand\Break{\algorithmicbreak}

\newcommand\ddfrac[2]{\frac{\displaystyle #1}{\displaystyle #2}}

\newtheorem{theorem}{Theorem}
\newtheorem{lemma}[theorem]{Lemma}

\theoremstyle{definition}

\newcommand\blfootnote[1]{%
  \begingroup
  \renewcommand\thefootnote{}\footnote{#1}%
  \addtocounter{footnote}{-1}%
  \endgroup
}

\settimeformat{hhmmsstime}
\mmddyyyydate
\title{
\Large\bf Scalable and Probabilistic Leaderless BFT Consensus through Metastability}
\author{%
\tronly{%
Team~Rocket,
Maofan~Yin,
Kevin~Sekniqi,
Robbert~van~Renesse,
Emin~G\"un~Sirer\\
Cornell University\textsuperscript{*}
}{%
Anonymous Submission 134
}
}
\date{}

\tronly{%
\usepackage[hang]{footmisc}
\setlength\footnotemargin{0em}
}

\begin{document}
\maketitle
\tronly{%
\blfootnote{%
\textsuperscript{*}\emph{Blasts off at the speed of light!} --- Team Rocket\\
An earlier version of this paper published on May 16th 2018, IPFS, was titled \emph{Snowflake to Avalanche: A Novel Metastable Consensus Protocol Family for Cryptocurrencies}.

}
}{}

\begin{abstract}
This paper introduces a family of leaderless Byzantine fault tolerance protocols, built around a metastable mechanism via network subsampling.
These protocols provide a strong probabilistic safety guarantee in the presence of Byzantine adversaries while their concurrent and leaderless nature enables them to achieve high throughput and scalability.
Unlike blockchains that rely on proof-of-work, they are quiescent and green.
Unlike traditional consensus protocols where one or more nodes typically process linear bits in the number of total nodes per decision, no node processes more than logarithmic bits. It does not require accurate knowledge of all participants and exposes new possible tradeoffs and improvements in safety and liveness for building consensus protocols.


The paper describes the Snow protocol family, analyzes its guarantees, and describes how it can be used to construct the core of an internet-scale electronic payment system called Avalanche, which is evaluated in a large scale deployment.
Experiments demonstrate that the system can achieve high throughput (3400 tps), provide low confirmation latency (1.35 sec), and scale well compared to existing systems that deliver similar functionality. For our implementation and setup, the bottleneck of the system is in transaction verification.
\end{abstract}

\section{Introduction}

Achieving agreement among a set of distributed hosts lies at the core of countless applications, ranging from Internet-scale services that serve billions of people~\cite{Burrows06,HuntKJR10} to cryptocurrencies worth billions of dollars~\cite{marketcapcryptocurrency}.
To date, there have been two main families of solutions to this problem.
Traditional consensus protocols 
rely on all-to-all communication to ensure that all correct nodes reach the same decisions with absolute certainty.
Because they usually require quadratic communication overhead and accurate knowledge of membership, they have been difficult to scale
to large numbers of participants.
On the other hand, Nakamoto consensus protocols~\cite{nakamoto2008bitcoin,GarayKL15, PassSS17, SompolinskyZ15, SompolinskyLZ16, SompolinskyZ18, BentovHMN17, EyalGSR16,Kokoris-KogiasJ16,PassS16a, PassS18} have become popular with the rise of Bitcoin.
These protocols provide a probabilistic safety guarantee: Nakamoto consensus decisions may revert with some probability $\varepsilon$.
A protocol parameter allows this probability to be rendered arbitrarily small, enabling high-value financial systems to be constructed on this foundation.
This family is a natural fit for open, permissionless settings where any node can join the system at any time.
Yet, these protocols are costly, wasteful, and limited in performance.
By construction, they cannot quiesce: their security relies on constant participation by miners, even when there are no decisions to be made.
Bitcoin currently consumes around 63.49 TWh/year~\cite{bitcoinpower}, about twice as all of Denmark~\cite{denmarkpower}.
Moreover, these protocols suffer from an inherent scalability bottleneck that is difficult to overcome through simple reparameterization~\cite{CromanDEGJKMSSS16}. 

This paper introduces a new family of consensus protocols called Snow.
Inspired by gossip algorithms, this family gains its properties through a deliberately metastable mechanism.
Specifically, the system operates by repeatedly sampling the network at random, and steering correct nodes towards a common outcome.
Analysis shows that this metastable mechanism is powerful: it can move a large network to an irreversible state quickly, where the irreversibility implies that a sufficiently large portion of the network has accepted a proposal and a conflicting proposal will not be accepted with any higher than negligible ($\varepsilon$) probability. 

Similar to Nakamoto consensus, the Snow protocol family provides a probabilistic safety guarantee, using a tunable security parameter that can render the possibility of a consensus failure arbitrarily small.
Unlike Nakamoto consensus, the protocols are green, quiescent and efficient; they do not rely on proof-of-work~\cite{DworkN92} and do not consume energy when there are no decisions to be made.
The efficiency of the protocols stems partly from removing the leader bottleneck: each node requires $\Oh{1}$ communication overhead per round and $\Oh{\log{n}}$ rounds in expectation, whereas classical consensus protocols have one or more nodes that require $\Oh{n}$ communication per round (phase).
Further, the Snow family tolerates discrepancies in knowledge of membership, as we discuss later. In contrast, classical consensus protocols require the full and accurate knowledge of $n$ as its safety foundation.

Snow's subsampled voting mechanism has two additional properties that improve on previous approaches for consensus. 
Whereas the safety of quorum-based approaches breaks down immediately when the predetermined threshold $f$ is exceeded,
Snow's probabilistic safety guarantee degrades smoothly when Byzantine participants exceed $f$.
This makes it easier to pick the critical threshold $f$.
It also exposes new tradeoffs between safety and liveness: the Snow family is more efficient when the fraction of Byzantine nodes is small, and it can be parameterized to tolerate more than a third of the Byzantine
nodes by trading off liveness.

To demonstrate the potential of this protocol family, we illustrate a practical
peer-to-peer payment system, Avalanche. In effect, Avalanche executes multiple Snowball (one from the Snow family) instances with the aid of a Directed Acyclic Graph (DAG). The DAG serves to piggyback multiple instances, reducing the cost from $\Oh{\log{n}}$ to $\Oh{1}$ per node and streamlining the path where there are no conflicting transactions.
%

Overall, the main contribution of this paper is to introduce a brand new family
of consensus protocols, based on randomized sampling and metastable decision.
The next section provides the model, goals and necessary assumptions for the new protocols.
Section~\ref{sec:protocol} gives intuition behind the protocols, followed by their full specification,
Section~\ref{sec:analysis} provides methodology used
by our formal analysis of safety and liveness in Appendix~\ref{sec:full-analysis},
Section~\ref{sec:implementation} describes Avalanche, a Bitcoin-like payment system,
Section~\ref{sec:evaluation} evaluates Avalanche,
Section~\ref{sec:related-work} presents related work, and finally, Section~\ref{sec:conclusions} summarizes our contributions.

\section{Model and Goals}
\label{sec:model_and_goals}

\tronly{\paragraph{Key Guarantees}}{}


\paragraph{Safety} Unlike classical consensus protocols, and similar to longest-chain-based consensus protocols such as Nakamoto consensus~\cite{nakamoto2008bitcoin}, we adopt an $\varepsilon$-safety guarantee that is probabilistic. 
In practice, this probabilistic guarantee is as strong as traditional safety guarantees, since appropriately small choices of $\varepsilon$ can render consensus failure negligible, lower than the probability of hardware failure due to random events.
Figure~\ref{fig:fandepsilon} shows how the portion ($f/n$) of misbehaving participants (or computation power) affects the probability of system safety failure (decision of two conflicting proposals), given a choice of finality.

\begin{figure}[h]
    \includegraphics[width=\linewidth]{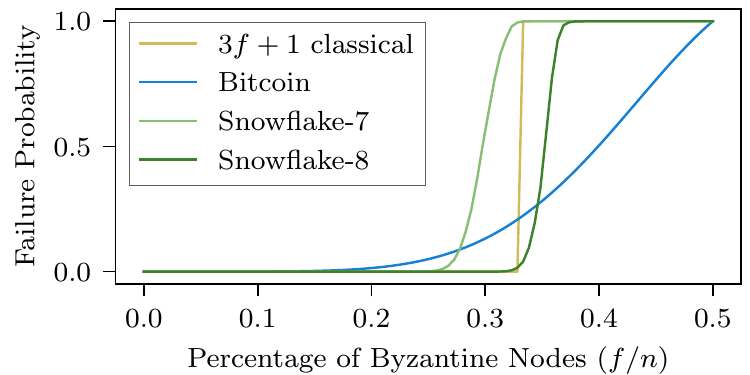}
    \caption{Classical BFT protocols that tolerate $f$ failures will encounter total safety failure when the threshold is exceeded even by one additional node. The Bitcoin curve shows a typical finality choice for Bitcoin where a block is considered final when it is ``buried'' in a branch having 6 additional blocks compared to any other competing forks. Snowflake belongs to the Snow family, and it is configured with $k=10$, $\beta=150$. Snowflake-7,8 uses $\alpha=7$ and $\alpha=8$ respectively (see Section~\ref{sec:protocol} for the definition of $k$, $\alpha$ and $\beta$.}
    \label{fig:fandepsilon}
\end{figure}

\tronly{%
\begin{figure}[h]
    \includegraphics[width=\linewidth]{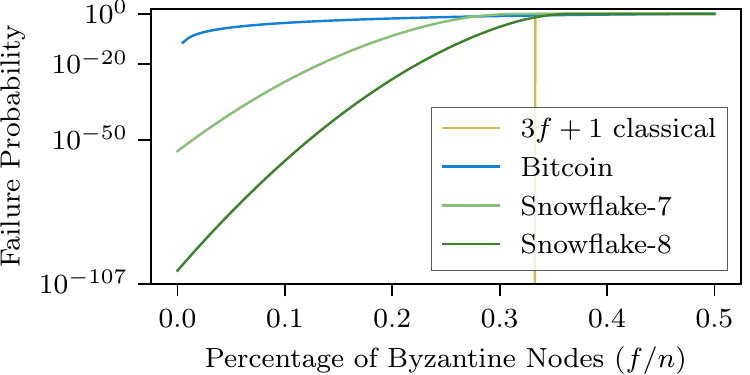}
    \caption{Figure~\ref{fig:fandepsilon} with log-scaled y-axis.}
    \label{fig:fandepsilonlog}
\end{figure}
}{}

\paragraph{Liveness} All our protocols provide a non-zero probability guarantee of termination within a bounded amount of time. 
This bounded guarantee is similar to various protocols such as Ben-Or~\cite{ben1983another} and longest-chain protocols.
In particular, for Nakamoto consensus, the number of required blocks for a transaction increases exponentially with the number of adversarial nodes, with an asymptote at $f = n/2$ wherein the number is infinite.
In other words, the time required for finality approaches $\infty$ as $f$ approaches $n/2$\tronly{ (Figure~\ref{fig:livenessproperties}).}{.}
Furthermore, the required number of rounds is calculable ahead of time, as to allow the system designer to tune liveness at the expense of safety. Lastly, unlike traditional consensus protocols and similar to Nakamoto, our protocols benefit from lower adversarial presence, as discussed in property P3 below.

\tronly{%
\begin{figure}[h!]
    \includegraphics[width=\linewidth]{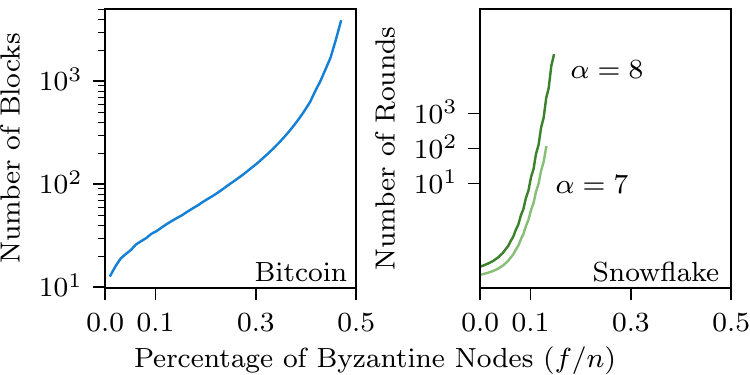}
    \caption{The relation between $f/n$ and the convergence speed, given $\varepsilon = 10^{-20}$. The left figure shows the expected number of blocks to guarantee $\varepsilon$ in Bitcoin, which, counter to commonly accepted folk wisdom, is not a constant $6$, but depends on adversary size to withhold the same $\varepsilon$. The right figure shows the maximum number of rounds required by Snowflake, where being different from Bitcoin, the asymptote is below $0.5$ and varies by the choice of parameters.}
    \label{fig:livenessproperties}
\end{figure}
}

\emph{Formal Guarantees}: Let the system be parameterized for an $\varepsilon$ safety failure probability under a maximum expected $f$ number of adversarial nodes. Let $\Oh{\log n} < t_{max} < \infty$ be the upper bound of the execution of the protocols. The Snow protocols then provide the following guarantees:
\begin{compactitem}

    \item \textbf{P1. Safety.} When decisions are made by any two correct nodes, they decide on conflicting transactions with negligible probability ($\leq \varepsilon$).
    \item \textbf{P2. Liveness (Upper Bound).} Snow protocols terminate with a strictly positive probability within $t_{max}$ rounds.  
    \item \textbf{P3. Liveness (Strong Form).} If $f \leq \Oh{\sqrt{n}}$, then the Snow protocols terminate with high probability ($\geq 1 - \varepsilon$) in $\Oh{\log{n}}$ rounds. 
\end{compactitem}

\paragraph{Network} 
In the standard definition of asynchrony~\cite{ben1983another}, message transmission time is finite, but the distribution is unspecified (and thus the delivery time can be unbounded for some messages). This implies that the scheduling of message transmission itself could behave arbitrarily, and potentially even maliciously (with full asynchrony).
We use a modified version of this model, which is well-accepted~\cite{banerjee2014epidemic, ganesh2005effect, draief2006thresholds, keeling2011modeling, liggett1997stochastic} in the analysis of epidemic networks and gossip-based stochastic systems. In particular, we fix the distribution of message delay to that of the exponential distribution.
We note that, just like in the standard asynchronous model, there is a strictly non-zero probability that any correct node may execute its next local round only after an arbitrarily large amount of time has passed.
Furthermore, we also note that scheduling only applies to correct nodes, and the adversary may execute arbitrarily, as discussed later. 

\paragraph{Achieving Liveness}
Classical consensus that works with asynchrony does not get stuck in a single phase of voting because the vote initiator always polls votes from all known participants and waits for $n - f$ responses.
In our system, however, nodes operate via subsampling, hence it is possible for a single sample to select a majority of adversarial nodes, and therefore the node gets stuck waiting for the responses. To ensure liveness, a node should be able to wait with some timeout. Therefore, our protocols are synchronous in order to guarantee liveness. Lastly, it is worth noting that Nakamoto consensus is synchronous, in which the required difficulty of proof-of-work is dependent on the maximum network delay~\cite{PassSS17}. 

\paragraph{Adversary}
The adversarial nodes execute under their own internal scheduler, which is unbounded in speed, meaning that all adversarial nodes can execute at any infinitesimally small point in time, unlike correct nodes. 
The adversary can view the state of every honest node at all times and can instantly modify the state of all adversarial nodes. 
It cannot, however, schedule or modify communication between correct nodes.
Finally, we make zero assumptions about the behavior of the adversary, meaning that it can choose any execution strategy of its liking.
In short, the adversary is computationally bounded (it cannot forge digital signatures) but otherwise is point-to-point informationally unbounded (knows all state) and round-adaptive (can modify its strategy at any time). 

\paragraph{Sybil Attacks}
Consensus protocols provide their guarantees based on assumptions that only a fraction of participants are adversarial.
These bounds could be violated if the network is naively left open to arbitrary participants.
In particular, a Sybil attack~\cite{douceur2002sybil}, wherein a large number of identities are generated by an adversary, could be used to exceed the bounds.

A long line of work, including PBFT~\cite{castro1999practical}, treats the Sybil problem separately from consensus, and rightfully so, as Sybil control mechanisms are distinct from the underlying, more complex agreement protocol\footnote{This is not to imply that every consensus protocol can be coupled/decoupled with every Sybil control mechanism.}.
In fact, to our knowledge, only Nakamoto-style consensus has ``baked-in'' Sybil prevention as part of its consensus, made possible by chained proof-of-work~\cite{aspnes2005exposing}, which requires miners to continuously stake a hardware investment.
Other protocols, discussed in Section~\ref{sec:related-work}, rely on proof-of-stake (by economic argument) or proof-of-authority (by administrative argument that makes the system ``permissioned'').
The consensus protocols presented in this paper can adopt any
Sybil control mechanism, although proof-of-stake is most aligned with their quiescent operation.
One can use an already established proof-of-stake based mechanism~\cite{GiladHMVZ17}.
The full deployment of an autonomous P2P payment system incorporating staking mechanism is beyond the scope of this paper, whose focus is on a novel design paradigm of the core consensus algorithm.


\paragraph{Flooding Attacks}
Flooding/spam attacks are a problem for any distributed system. 
Without a protection mechanism, an attacker can generate large numbers of transactions and flood protocol data structures, consuming storage.
There are a multitude of techniques to deter such attacks, including network-layer protection, proof-of-authority, local proof-of-work and economic mechanisms. 
In Avalanche, we use transaction fees, making such attacks costly even if the attacker is sending money back to addresses under its control.

\paragraph{Additional Assumptions}
We do not assume that all members of the network are known to all participants, but rather may temporarily have some discrepancies in network view.
We quantify the bounds on the discrepancy in Appendix~\ref{sec:full-analysis-churn}.
We assume a safe bootstrapping mechanism, similar to that of Bitcoin, that enables a node to connect with sufficiently many correct nodes to acquire a statistically unbiased view of the network.
We do not assume a PKI\@.
Finally, we make standard cryptographic assumptions related to digital signatures and hash functions.

\section{Protocol Design}\label{sec:protocol}
We start with a non-BFT protocol called Slush and progressively build up to Snowflake and Snowball, all based on the same common majority-based metastable voting mechanism.
These protocols are single-decree consensus protocols of increasing robustness.
We provide full specifications for the protocols in this section, and defer the analysis to the next section, and present formal proofs in the appendix.

\subsection{Slush: Introducing Metastability}
The core of our approach is a single-decree consensus protocol,
inspired by epidemic or gossip protocols.
%
The simplest protocol, Slush, is the foundation of this family, shown in Figure~\ref{fig:slush-loop}.
Slush is \emph{not} tolerant to Byzantine faults, only crash-faults (CFT), but serves as an illustration for the BFT protocols that follow.
For ease of exposition, we will describe the operation of Slush using a decision between two conflicting colors, red and blue.

In Slush, a node starts out initially in an uncolored state.
Upon receiving a transaction from a client, an uncolored node updates its own color to the one carried in the transaction and initiates a query.
To perform a query, a node picks a small, constant sized ($k$) sample of the network uniformly at random, and sends a query message.
Upon receiving a query, an uncolored node adopts the color in the query, responds with that color, and initiates its own query, whereas a colored node simply responds with its current color.
Once the querying node collects $k$ responses, it checks if a fraction $\ge\alpha$ are for the same color, where $\alpha > \floor{k/2}$ is a protocol parameter.
If the $\alpha$ threshold is met and the sampled color differs from the node's own color, the node flips to that color.
It then goes back to the query step, and initiates a subsequent round of query, for a total of $m$ rounds.
Finally, the node decides the color it ended up with at time $m$.

\algnewcommand{\IIf}[1]{\State\algorithmicif\ #1\ \algorithmicthen}
\algnewcommand{\EndIIf}{\unskip\ \algorithmicend\ \algorithmicif}

\newcommand{\codecolor}{\mathit{col}}
\newcommand{\codecount}{\mathit{cnt}}
\newcommand{\codelastcol}{\mathit{lastcol}}
\newcommand{\assign}{\coloneqq}
\begin{figure}
    \small
\begin{algorithmic}[1]
    \Procedure{onQuery}{$v, \codecolor'$}
        \IIf{$\codecolor = \bot$} $\codecolor \assign \codecolor'$
        \State\Call{respond}{$v, \codecolor$}
    \EndProcedure
    \Procedure{slushLoop}{$u, \codecolor_0 \in \{\texttt{R}, \texttt{B}, \bot\}$}
        \State $\codecolor \assign \codecolor_0$ \textrm{// initialize with a color}
        \For{$r \in \{1\ldots m\}$}
            \State \textrm{// if $\bot$, skip until \textsc{onQuery} sets the color}
            \If{$\codecolor = \bot$}
                \Continue
            \EndIf
            \State \textrm{// randomly sample from the known nodes}
            \State $\mathcal{K} \assign \Call{sample}{\mathcal{N}\backslash u, k}$
            \State $P \assign \texttt{[}\Call{query}{v, \codecolor}\quad\textbf{for}\ v \in \mathcal{K}\texttt{]}$
            \For{$\codecolor' \in \{\texttt{R}, \texttt{B}\}$}
                \If{$P.\Call{count}{\codecolor'} \ge \alpha$}
                    \State $\codecolor \assign \codecolor'$
                \EndIf
            \EndFor
        \EndFor
        \State \Call{accept}{$\codecolor$}
    \EndProcedure
    \captionof{figure}{Slush protocol. Timeouts elided for readability.}\label{fig:slush-loop}
\end{algorithmic}
\end{figure}
 
Slush has a few properties of interest. 
First, it is almost \emph{memoryless}: a node retains no state between rounds other than its current color, and in particular maintains no history of interactions with other peers.
Second, unlike traditional consensus protocols that query every participant, every round involves sampling just a small, constant-sized slice of the network at random.
Third, Slush makes progress under any network configuration (even fully bivalent state, i.e. 50/50 split between colors), since random perturbations in sampling will cause one color to gain a slight edge and repeated samplings afterwards will build upon and amplify that imbalance.
Finally, if $m$ is chosen high enough, Slush ensures that all nodes will be colored identically with high probability (whp).
Each node has a constant, predictable communication overhead per round, and $m$ grows logarithmically with $n$.

\tronly{
The Slush protocol does not provide a strong safety guarantee in the presence of Byzantine nodes.
In particular, if the correct nodes develop a preference for one color, a Byzantine adversary can attempt to flip nodes
to the opposite so as to keep the network in balance, preventing a decision.
We address this in our first BFT protocol that introduces more state storage at the nodes.
}{
We next examine how to extend Slush to tolerate Byzantine behavior.
}

\subsection{Snowflake: BFT}

Snowflake augments Slush with a single counter that captures the strength of a node's conviction in its current color.
This per-node counter stores how many consecutive samples of the network by that node have all yielded the same color.
A node accepts the current color when its counter reaches $\beta$, another security parameter.
Figure~\ref{fig:snowflake-loop} shows the amended protocol, which includes
the following modifications:

\begin{compactenum}
	\item Each node maintains a counter $\mathit{cnt}$;
    \item Upon every color change, the node resets $\mathit{cnt}$ to 0;
    \item Upon every successful query that yields $\ge \alpha$ responses for the same color as the node, the node increments $\mathit{cnt}$.
\end{compactenum}

\newcommand{\codemaj}{\mathit{maj}}
\newcommand{\codefalse}{\texttt{false}}
\newcommand{\codetrue}{\texttt{true}}
\begin{figure}
    \small
\begin{algorithmic}[1]
    \Procedure{snowflakeLoop}{$u, \codecolor_0 \in \{\texttt{R}, \texttt{B}, \bot\}$}
        \State $\codecolor \assign \codecolor_0$, $\codecount \assign 0$
        \While{\textrm{undecided}}
            \If{$\codecolor = \bot$}
                \Continue
            \EndIf
            \State $\mathcal{K} \assign \Call{sample}{\mathcal{N}\backslash u, k}$
            \State $P \assign \texttt{[}\Call{query}{v, \codecolor}\quad\textbf{for}\ v \in \mathcal{K}\texttt{]}$
            \State $\codemaj \assign \codefalse$
            \For{$\codecolor' \in \{\texttt{R}, \texttt{B}\}$}
            \If{$P.\Call{count}{\codecolor'} \ge \alpha$}
                \State $\codemaj \assign \codetrue$
                \If{$\codecolor' \neq \codecolor$}
                    \State $\codecolor \assign \codecolor'$, $\codecount \assign 1$
                \Else\hspace{1ex}$\codecount\texttt{++}$
                \EndIf
                \IIf {$\codecount \ge \beta$} \Call{accept}{$\codecolor'$}
            \EndIf
            \EndFor
            \IIf{$\codemaj = \codefalse$} $\codecount \assign 0$
        \EndWhile
    \EndProcedure
    \captionof{figure}{Snowflake.}\label{fig:snowflake-loop}
\end{algorithmic}
\end{figure}

\begin{figure}[t]
    \small
\begin{algorithmic}[1]
    \Procedure{snowballLoop}{$u, \codecolor_0 \in \{\texttt{R}, \texttt{B}, \bot\}$}
        \State $\codecolor \assign \codecolor_0$, $\codelastcol \assign \codecolor_0$, $\codecount \assign 0$
        \State $d[\texttt{R}] \assign 0$, $d[\texttt{B}] \assign 0$
        \While{\textrm{undecided}}
            \IIf{$\codecolor = \bot$} \Continue
            \State $\mathcal{K} \assign \Call{sample}{\mathcal{N}\backslash u, k}$
            \State $P \assign \texttt{[}\Call{query}{v, \codecolor}\quad\textbf{for}\ v \in \mathcal{K}\texttt{]}$
            \State $\codemaj \assign \codefalse$
            \For{$\codecolor' \in \{\texttt{R}, \texttt{B}\}$}
            \If{$P.\Call{count}{\codecolor'} \ge \alpha$}
                \State $\codemaj \assign \codetrue$
                \State $d[\codecolor']\texttt{++}$
                \If {$d[\codecolor'] > d[\codecolor]$}
                        \State$\codecolor \assign \codecolor'$
                \EndIf
                \If{$\codecolor' \neq \codelastcol$}
                    \State$\codelastcol \assign \codecolor'$, $\codecount \assign 1$
                \Else\hspace{1ex}$\codecount\texttt{++}$
                \EndIf
                \IIf {$\codecount \ge \beta$} \Call{accept}{$\codecolor'$}
            \EndIf
            \EndFor
            \IIf{$\codemaj = \codefalse$} $\codecount \assign 0$
        \EndWhile
    \EndProcedure
    \captionof{figure}{Snowball.}\label{fig:snowball-loop}
\end{algorithmic}
\end{figure}

When the protocol is correctly parameterized for a given threshold of Byzantine nodes and a desired $\varepsilon$-guarantee, it can ensure both safety (P1) and liveness (P2, P3).
As we later show, there exists an irreversible state after which a decision is inevitable. Correct nodes begin to commit past the irreversible state to adopt the same color, whp. For additional intuition, which we do not expand in this paper, there also exists a phase-shift point, where the Byzantine nodes lose ability to keep network in a bivalent state.

\subsection{Snowball: Adding Confidence}

Snowflake's notion of state is ephemeral: the counter gets reset with every color flip.
Snowball augments Snowflake with \emph{confidence counters} that capture the number of queries that have yielded a threshold result for their corresponding color (Figure ~\ref{fig:snowball-loop}).
A node decides if it gets $\beta$ consecutive chits for a color. However, it only changes preference based on the total accrued confidence.
The differences between Snowflake and Snowball are as follows:
\begin{compactenum}
    \item Upon every successful query, the node increments its confidence counter for that color.
    \item A node switches colors when the confidence in its current color becomes lower than the confidence value of the new color.
\end{compactenum}

\section{Analysis}
\label{sec:analysis}
Due to space limits, we move some core details to Appendix~\ref{sec:full-analysis}, where we show that under certain independent and distinct assumptions, the Snow family of consensus protocols provide safety (P1) and liveness (P2, P3) properties.
In this section, we summarize our core results and provide some proof sketches.

\paragraph{Notation} Let the network consist of a set of $n$ nodes (represented by set $\mathcal{N}$), where $c$ are correct nodes (represented by set $\mathcal{C}$) and $f$ are Byzantine nodes (represented by set $\mathcal{B}$). 
Let $u, v \in \mathcal{C}$ refer to any two correct nodes in the network. Let $k, \alpha, \beta \in \mathbb{Z}^+$ be positive integers where $\alpha > \floor{k/2}$. From now on, $k$ will always refer to the network sample size, where $k \leq n$, and $\alpha$ will be the majority threshold required to consider the voting experiment a ``success''. In general, we will refer to $\mathcal{S}$ as the state (or configuration) of the network at any given time. 

\paragraph{Modelling Framework} To formally model our protocols, we use continuous-time Markov processes (CTMC). 
The state space is enumerable (and finite), and state transitions occur in continuous time. 
CTMCs naturally model our protocols since state transitions do not occur in epochs and in lockstep for every node (at the end of every time unit) but rather occur at any time and independently of each other. 

We focus on binary consensus, although the safety results generalize to more than two values. We can think of the network as a set of nodes either colored red or blue, and we will refer to this configuration at time $t$ as $\mathcal{S}_t$. 
We model our protocols through a continuous-time process with two absorbing states, where either all nodes are red or all nodes are blue. 
The state space $\mathcal{S}$ of the stochastic process is a condensed version of the full configuration space, where each state $\{0, \dots, n\}$ represents the total number of blue nodes in the system. 

The simplification that allows us to analyze this system is to obviate the need to keep track of all of the execution paths, as well as all possible adversarial strategies, and rather focus entirely on a single state of interest, without regards to how we achieve this state. 
More specifically, the core extractable insight of our analysis is in identifying the \textit{irreversibility} state of the system, the state upon which so many correct nodes have usurped either red or blue that reverting back to the minority color is highly unlikely. 

\subsection{Safety}

\paragraph{Slush} 
We assume that
all nodes share the same $\mathcal{N}$, and in
Appendix~\ref{sec:full-analysis-churn}, we sketch how to relax the requirement of the membership knowledge.
We model the dynamics of the system through a continuous-time process where two states are absorbing, namely the all-red or all-blue state. Let $\{X_{t \geq 0}\}$ be the random variable that describes the state of the system at time $t$, where $X_0 \in \{0, \dots, c\}$.
We begin by immediately discussing the most important result of the safety dynamics of our processes: the \emph{reversibility} probabilities of the \textbf{Slush} process. All the other formal results in this paper are, informally speaking, intuitive derivations and augmentations of this result. 

\begin{theorem}
Let the configuration of the system at time $t$ be $\mathcal{S}_t = n/2 + \delta$, meaning that the network has drifted to $2\delta$ more blue nodes than red nodes ($\delta = 0$ means that red and blue are equal). Let $\xi_\delta$ be the probability of absorption to the all-red state (minority). Then, for all $0 \leq \delta \leq n/2$, we have 
\begin{equation}
\begin{split}
    \xi_\delta &\leq \left(\dfrac{1/2 - \delta/n}{\alpha/k}\right)^{\alpha}\left(\dfrac{1/2 + \delta/n}{1- \alpha/k}\right)^{k-\alpha}\\
    &\leq e^{-2((\alpha/k) - (1/2) + (\delta/n))^2 k}
\end{split}
\end{equation}
\end{theorem}

\begin{proof}
This bound follows from the Hoeffding-derived tail bounds of the hypergeometric distribution by Chvatal~\cite{chvatal1979tail}. 
\end{proof}

We note that Chvatal's bounds are introduced for simplicity of exposition and are extremely weak.  
We leave the full closed-form expression in Theorem~\ref{theorem:slush_prob_convergence_minority} to the appendix, which is also significantly stronger than the Chvatal bound. 
Nonetheless, using the loose Chvatal bound, we make the key observation that as the drift $\delta$ increases, given fixed $\alpha$ and $k$, the probability of moving towards the minority value decreases \emph{exponentially fast} (in fact, even faster, since there is a quadratic term in the inverse exponent). Additionally, the same result holds for increasing $\alpha$ given a fixed $k$. 

The outcomes of this theorem demonstrate a key property: once the network loses full bivalency (i.e. $\delta > 0$), it tends to topple and converge rapidly towards the majority color, unable to revert back to the minority with significant probability. This is the fundamental property exploited by our protocols, and what makes them secure despite only sampling a small, constant-sized set of the network. The core result that follows for the safety guarantees in Snowflake is in finding regions (given specific parameter choices) where the reversibility holds with no higher than $\varepsilon$ probability even under adversarial presence. 

\paragraph{Snowflake} 
For Snowflake, we assume that some fraction of nodes are adversarial. In Slush, once the network gains significant majority support for one proposal (e.g., the color blue), it becomes unlikely for a minority proposal (e.g., the color red) to ever become decided in the future (irreversibility). Furthermore, in Slush nodes simply have to execute the protocol for a deterministic number of rounds, $m$, which is known ahead of protocol execution. When introducing adversarial nodes with arbitrary strategies, however, nodes cannot simply execute the protocol for a deterministic number of rounds, since the adversary may nondeterministically affect the value of $m$. Instead, correct nodes must implement a mechanism to explicitly detect that irreversibility has been reached. To that end, in Snowflake, every correct node implements a decision function, $\mathcal{D}(u, \mathcal{S}_t, blue) \rightarrow \{0, 1\}$, which is a random variable that outputs $1$ if node $u$ detects that the network has reached an irreversibility state at time $t$ for blue. The decision mechanism is probabilistic, meaning that it can fail, although it is designed to do so with negligible probability. We now sketch the proof of Snowflake.

\noindent \emph{Proof Sketch}. We define safety failure to be the event wherein any two correct nodes $u$ and $v$ decide on blue and red, i.e. $\mathcal{D}(u, \mathcal{S}_t, blue) \rightarrow 1$ and $\mathcal{D}(v, \mathcal{S}_{t'}, red) \rightarrow 1$, for any two times $t$ and $t'$. We again model the system as a continuous time random process. The state space is defined the same way as in Slush. However, we note some critical subtleties. First, even if all correct nodes accept a color, the Byzantine nodes may revert. Second, we also have to consider the decision mechanism $\mathcal{D}(*)$. To analyze, we obviate the need to keep track of all network configurations under all adversarial strategies and assume that a node $u$ first decides on blue. Then, conditioned on the state of the network upon $u$ deciding, we calculate the probability that another node $v$ decides red, which is a function of both the probability that the network reverts towards a minority blue state and that $v$ decides at that state. 
We show that under appropriate choices of $k$, $\alpha$, and $\beta$, we can construct highly secure instances of Snowflake (i.e. safety failure with probability $\leq \varepsilon$) when the network reaches some bias of $\delta$, as shown in Figure~\ref{fig:states_feasible_solutions}. A concrete example is provided in Figure~\ref{fig:fandepsilon}.

\begin{figure}[h]
\begin{center}
\usetikzlibrary{decorations.pathreplacing}
\begin{tikzpicture}[x=1.12cm, scale=0.8, every node/.append style={transform shape}]
    \draw[line width=0.2ex, line cap=round] (0,0) -- (8,0); 
    \foreach \x in {0,4,6.5,8} 
    \draw[line width=0.2ex, line cap=round, shift={(\x,0)},color=black] (0pt,3pt) -- (0pt,-3pt);
    \foreach \x/\y/\z in {%
        0/$0$/a,
        4/${c/2}$/b,
        6.5/$$/c,
        8/$c$/d}
    \draw[line width=0.2ex, shift={(\x,0)},color=black] (0pt,0pt) -- (0pt,-3pt) node[below] (\z) {\y};

    \begin{scope}[line width=0.15ex]
        \path[->] (c) edge[out=-90, in=-90, looseness=0.4] node[above] {$\le \varepsilon$} (a);
    \end{scope}
    \draw [thick,decoration={brace,raise=1ex, amplitude=2ex},decorate] (4, 0) -- (6.5, 0) node[midway,above,yshift=3.5ex] {$\delta$};
\end{tikzpicture}
\caption{Representation of the irreversibility state, which exists when -- even under $f$ Byzantine nodes -- the number of blue correct nodes exceeds that of red correct nodes by more than $2\delta$.
}
\label{fig:states_feasible_solutions}
\end{center}
\end{figure}
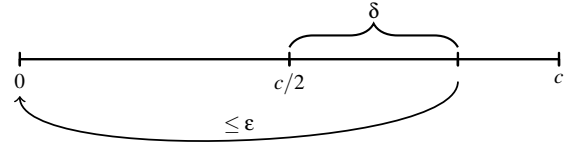    


\paragraph{Snowball}
Snowball is an improvement in security over Snowflake, where random perturbations in network samples are reduced by introducing a limited form of history, which we refer to as confidence. 

\noindent \emph{Proof Sketch}. We apply martingale concentration inequalities to prove that once the system has reached the irreversibility state, then the growth of the confidence of the majority decided color will perpetually grow and drift further away from those of the minority color, effectively rendering reversibility less likely over time. If the drifts ever revert, then reversibility analysis becomes identical to that of Snowflake.  

\subsection{Liveness}

We assume that the observed adversarial presence $0 \leq f' \leq n(k - \alpha - \psi)/k \leq f$, where we refer to $\psi$ as the buffer zone. 
The bigger $\psi$, the quicker the ability of the decision mechanism to finalize a value. 
If, of course, $\psi$ approaches zero or becomes negative, then we violate the upper bound of adversarial tolerance for the parameterized system, and thus the adversary can, with high probability, stall termination by simply choosing to not respond, although the safety guarantees may still hold. 

Assuming that $\psi$ is strictly positive, termination is strictly finite under all network configurations where a proposal has at least $\alpha$ support. Furthermore, not only is termination finite with probability one, we also have a strictly positive probability of termination within any bounded amount of time $t_{max}$, as discussed in Lemma~\ref{lemma:finitetermination}, which follows from Theorem~\ref{theorem:mean-convergence-time}. This captures liveness property P2. 

\noindent\emph{Proof Sketch.} Using the construction of the system to prove irreversibility, we characterize the distribution of the average time spent (sojourn times) at each state before the system terminates execution by absorption at either absorbing state. The termination time is then a union of these times. 

For non-conflicting transactions, since the adversary is unable to forge a conflict, the time to decision is simply the mixing time , which is $\Oh{\log{n}}$.
Liveness guarantees under a fully bivalent network configuration reduce to an optimal convergence time of $\Oh{\log{n}}$ rounds if the adversary is at most $\Oh{\sqrt{n}}$, for $\alpha = \floor{k/2} + 1$. We leave additional detains to Lemma~\ref{lemma:centrallimit}.
When the adversary surpasses $\Oh{\sqrt{n}}$ nodes, the worst-case number of rounds increases polynomially, and as $f$ approaches $n/2$ it approaches exponential convergence rates.

\noindent\emph{Proof Sketch.} We modify Theorem~\ref{theorem:mean-convergence-time} to include the adversary, which reverts any imbalances in the network by keeping network fully bivalent. 

\paragraph{Multi-Value Consensus}
Our binary consensus protocol could support multi-value consensus by running logarithmic binary instances, one for each bit of the proposed value. However, such theoretical reduction might not be efficient in practice. Instead, we could directly incorporate multi-values as multi-colors in the protocol, where safety analysis could still be generalized.

For liveness, we sketch a leaderless initialization mechanism, which in expectation uses $\Oh{\log{n}}$ rounds under the assumption that the network is synchronized in the Appendix~\ref{sec:sync-heuristic}. 
While the design of initialization mechanisms is interesting, note that it is not necessary for a decentralized payment system, as we show in Section~\ref{sec:implementation}.
Finally, in the Appendix~\ref{sec:full-analysis-churn}, we discuss churn and view discrepancies.

\section{Peer-to-Peer Payment System}
\label{sec:implementation}


Using Snowball consensus, we have implemented a bare-bones payment system, Avalanche, which supports Bitcoin transactions. In this section, we describe the design and sketch how the implementation can support the value transfer primitive at the center of cryptocurrencies.
Deploying a full cryptocurrency involves bootstrapping, minting, staking, unstaking,
and inflation control. While we have solutions for these issues, their full
discussion is beyond the scope of this paper, whose focus is centered on the
novel Snow consensus protocol family.

In a cryptocurrency setting, cryptographic signatures enforce that only a key owner is able to create a transaction that spends a particular coin. Since correct clients follow the protocol as prescribed and never double spend coins, in Avalanche, they are guaranteed both safety and liveness for their \emph{virtuous} transactions. In contrast, liveness is not guaranteed for \emph{rogue} transactions, submitted by Byzantine clients, which conflict with one another. Such decisions may stall in the network, but have no safety impact on virtuous transactions.
We show that this is a sensible tradeoff, and that the resulting system is sufficient for building complex payment systems.

\subsection{Avalanche: Adding a DAG}

Avalanche consists of multiple single-decree Snowball instances instantiated as a multi-decree protocol that
maintains a dynamic, append-only directed acyclic graph (DAG) of all known transactions.
The DAG has a single sink that is the \emph{genesis vertex}.
Maintaining a DAG provides two significant benefits.
First, it improves efficiency, because a single vote on a DAG vertex implicitly votes for all transactions on the path to the genesis vertex.
Second, it also improves security, because the DAG intertwines the fate of transactions, similar to the Bitcoin blockchain.
This renders past decisions difficult to undo without the approval of correct nodes.

When a client creates a transaction, it names one or more \emph{parents}, which are included inseparably in the transaction and form the edges of the DAG\@.
The parent-child relationships encoded in the DAG may, but do not need to, correspond to application-specific dependencies; for instance, a child transaction need not spend or have any relationship with the funds received in the parent transaction.
We use the term \emph{ancestor set} to refer to all transactions reachable via parent edges back in history, and \emph{progeny} to refer to all children transactions and their offspring.

The central challenge in the maintenance of the DAG is to choose among \emph{conflicting transactions}.
The notion of conflict is application-defined and transitive, forming an equivalence relation.
In our cryptocurrency application, transactions that spend the same funds (\emph{double-spends}) conflict, and form a \emph{conflict set}
(shaded regions in Figure~\ref{fig:dag-cd}), out of which only a single one can be accepted.
Note that the conflict set of a virtuous transaction is always a singleton.
 

Avalanche instantiates a Snowball instance for each conflict set.
Whereas Snowball uses repeated queries and multiple counters to capture the amount of confidence built in conflicting transactions (colors),
Avalanche takes advantage of the DAG structure and uses a transaction's progeny.
Specifically, when a transaction $T$ is queried, all transactions reachable from $T$ by following the DAG edges are implicitly part of the query.
A node will only respond positively to the query if $T$ and its entire ancestry are currently the preferred option in their respective conflict sets.
If more than a threshold of responders vote positively, the transaction is said to collect a \emph{chit}.
Nodes then compute their \emph{confidence} as the total number of chits in the progeny of that transaction.
They query a transaction just once and rely on new vertices and possible chits, added to the progeny, to build up their confidence.
Ties are broken by an initial preference for first-seen transactions.
Note that chits are decoupled from the DAG structure, making the protocol immune to attacks where
the attacker generates large, padded subgraphs.

\subsection{Avalanche: Specification}
\label{subsection:specification}

\newcommand{\codeedges}{\mathit{edges}}
\newcommand{\codedata}{\mathit{data}}
\begin{figure}
\begin{center}
\small
\begin{algorithmic}[1]
    \Procedure{init}{}
        \State $\mathcal{T} \assign \emptyset$ \hspace{1ex}\textrm{// the set of known transactions}
        \State $\mathcal{Q} \assign \emptyset$ \hspace{1ex}\textrm{// the set of queried transactions}
    \EndProcedure
    \Procedure{onGenerateTx}{$\codedata$}
        \State $\codeedges \assign \{T' \gets T: T' \in \Call{parentSelection}{\mathcal{T}}\}$
        \State $T \assign \Call{Tx}{\codedata, \codeedges}$
        \State \Call{onReceiveTx}{$T$}
    \EndProcedure
    \Procedure{onReceiveTx}{$T$}
        \If{$T \notin \mathcal{T}$}
            \If{$\mathcal{P}_T = \emptyset$}
                \State $\mathcal{P}_T \assign \{T\}$, $\mathcal{P}_T\mathit{.pref} \assign T$
                \State $\mathcal{P}_T\mathit{.last} \assign T, \mathcal{P}_T\mathit{.cnt} \assign 0$
            \Else$\ \mathcal{P}_T \assign \mathcal{P}_T \cup \{T\}$
            \EndIf
            \State $\mathcal{T} \assign \mathcal{T} \cup \{T\}$, $c_T \assign 0$.
        \EndIf
    \EndProcedure
    \captionof{figure}{Avalanche: transaction generation.}\label{fig:gossipchain-ongen}
\end{algorithmic}
\end{center}
\end{figure}

\begin{figure}
\begin{center}
\small
\begin{algorithmic}[1]
    \Procedure{avalancheLoop}{}
        \While {\codetrue}
            \State\textrm{find  $T$ that satisfies }
            $T \in \mathcal{T} \land T \notin \mathcal{Q}$
            \State $\mathcal{K} \assign \Call{sample}{\mathcal{N}\backslash u, k}$
            \State $P \assign \sum_{v \in \mathcal{K}}\Call{query}{v, T}$
            \If{$P \ge \alpha$}
                \State $c_T \assign 1$
            \State\textrm{// update the preference for ancestors}
            \For{$T' \in \mathcal{T} : T' \stackrel{*}{\gets} T$}
                \If{$d(T') > d(\mathcal{P}_{T'}\mathit{.pref})$}
                    \State $\mathcal{P}_{T'}\mathit{.pref} \assign T'$
                \EndIf
                \If{$T'\neq \mathcal{P}_{T'}\mathit{.last}$}
                    \State $\mathcal{P}_{T'}\mathit{.last} \assign T'$, $\mathcal{P}_{T'}\mathit{.cnt} \assign 1$
                \Else
                    \State \texttt{++}$\mathcal{P}_{T'}\mathit{.cnt}$
                \EndIf
            \EndFor
            \Else
            \For{$T' \in \mathcal{T} : T' \stackrel{*}{\gets} T$}
                    \State$\mathcal{P}_{T'}\mathit{.cnt} \assign 0$
            \EndFor
            \EndIf
            \State\textrm{// otherwise, }$c_T$\textrm{ remains 0 forever}
            \State $\mathcal{Q} \assign \mathcal{Q} \cup \{T\}$ \hspace {1ex} \textrm{// mark T as queried}
        \EndWhile
    \EndProcedure
    \captionof{figure}{Avalanche: the main loop.}\label{fig:gossipchain-main}
\end{algorithmic}
\end{center}
\end{figure}

\begin{figure}[t]
\begin{center}
\small
\begin{algorithmic}[1]
    \Function{isPreferred}{$T$}
        \State \Return $T = \mathcal{P}_T\mathit{.pref}$
    \EndFunction
    \Function{isStronglyPreferred}{$T$}
        \State \Return $\forall T'\in\mathcal{T}, T' \stackrel{*}{\gets} T: \Call{isPreferred}{T'}$
    \EndFunction
    \Function{isAccepted}{$T$}
        \State\Return
            \vspace*{-.5\baselineskip}
        \begin{align*}
            (&(\forall T' \in \mathcal{T}, T' \gets T: \Call{isAccepted}{T'}) \\
                &\land |\mathcal{P}_T| = 1 \land \mathcal{P}_T\mathit{.cnt} \ge \beta_1) \texttt{\hspace{.1in}// safe early commitment} \\
            \lor &(
            \mathcal{P}_T\mathit{.cnt} \ge \beta_2)\texttt{\hspace{.1in}// consecutive counter}
        \end{align*}
    \EndFunction
    \Procedure{onQuery}{$j, T$}
        \State \Call{onReceiveTx}{$T$}
        \State \Call{respond}{$j, \textsc{isStronglyPreferred}(T)$}
    \EndProcedure
    \captionof{figure}{Avalanche: voting and decision primitives.}\label{fig:gossipchain-onquery}
\end{algorithmic}
\end{center}
\end{figure}

Each correct node $u$
keeps track of all transactions it has learned about in set $\mathcal{T}_u$,
partitioned into mutually exclusive conflict sets $\mathcal{P}_T$, $T \in \mathcal{T}_u$.
Since conflicts are transitive, if $T_i$ and $T_j$ are conflicting, then they belong to the same conflict set, i.e. $\mathcal{P}_{T_i} = \mathcal{P}_{T_j}$. This relation may sound counter-intuitive: conflicting transitions have the \emph{equivalence} relation, because they are equivocations spending the \emph{same} funds.

We write $T' \gets T$ if $T$ has a parent edge to transaction $T'$,
The ``$\stackrel{*}{\gets}$''-relation is its reflexive transitive closure, indicating a path from $T$ to $T'$.
DAGs built by different nodes are guaranteed to be compatible, though at any one time, the two nodes may not have a complete view of all vertices in the system.
Specifically, if $T' \gets T$, then every node in the system that has $T$ will also have $T'$ and the same relation $T' \gets T$; and conversely, if $T' \cancel{\gets} T$, then no nodes will end up with $T' \gets T$.

Each node $u$ can compute a confidence value, $d_u(T)$, from the progeny as follows:
\[ d_u(T) = \sum_{T' \in \mathcal{T}_u, T \stackrel{*}{\gets} T'}c_{uT'}\]
where $c_{uT'}$ stands for the chit value of $T'$ for node $u$. Each transaction initially has a chit value of $0$ before the node gets
the query results. If the node collects a threshold of $\alpha$ yes-votes after the query, the value $c_{uT'}$ is set to 1, otherwise remains $0$ forever.
Therefore, a chit value reflects the result from the one-time query of its associated transaction and becomes immutable afterwards, while $d(T)$ can increase as the DAG grows by collecting more chits in its progeny.
Because $c_T \in \{0, 1\}$, confidence values are monotonic.

In addition, node $u$ maintains its own local list of known nodes $\mathcal{N}_u \subseteq \mathcal{N}$ that comprise the system.
For simplicity, we assume for now $\mathcal{N}_u = \mathcal{N}$, and elide subscript $u$ in contexts without ambiguity.
%

Each node implements an event-driven state machine, centered around a query that serves both to solicit votes on each transaction and to notify other nodes of the existence of newly discovered transactions.
In particular, when node $u$ discovers a transaction $T$ through a query, it starts a one-time query process by sampling $k$ random peers and sending a message to them, after $T$ is delivered via $\textsc{onReceiveTx}$.

Node $u$ answers a query by checking whether each $T'$ such that $T' \stackrel{*}{\gets} T$ is currently preferred among competing transactions $\forall T'' \in \mathcal{P}_{T'}$.
If every single ancestor $T'$ fulfills this criterion, the transaction is said to be \emph{strongly preferred}, and receives a yes-vote (1). A failure of this criterion at any $T'$ yields a no-vote (0).
When $u$ accumulates $k$ responses, it checks whether there are $\alpha$ yes-votes for $T$, and if so grants the chit (chit value $c_T \assign 1$) for $T$.
The above process will yield a labeling of the DAG with a chit value and associated confidence for each transaction~$T$.

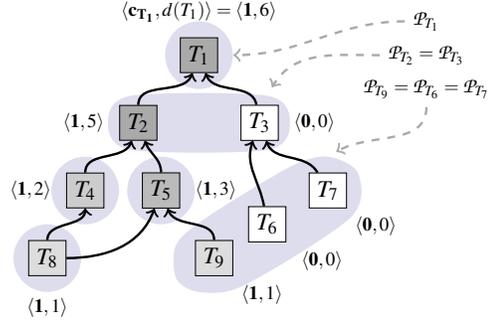
\begin{figure}
\begin{center}
    \definecolor{lightgray}{HTML}{dddddd}
\definecolor{medgray}{HTML}{cccccc}
\definecolor{medgray2}{HTML}{bbbbbb}
\definecolor{darkgray}{HTML}{aaaaaa}
\definecolor{lightgp}{HTML}{ddddee}
\begin{tikzpicture}[x=1.12cm, scale=0.9, every node/.append style={transform shape}]
    \begin{scope}[all/.style={draw, minimum height=0.5cm, minimum width=0.5cm},line width=0.08ex]
        \node[fill=lightgp, circle,minimum width=1cm] (pt1) at (0, 0) {};
        \node[fill=lightgp, circle,minimum width=1cm] (pt4) at (-1.5, -2) {};
        \node[fill=lightgp, circle,minimum width=1cm] (pt5) at (-0.5, -2) {};
        \node[fill=lightgp, circle,minimum width=1cm] (pt6) at (-2, -3) {};
        \fill[color=lightgp] plot[smooth cycle] coordinates { (1.2, -1) (0.9, -1.4) (-0.9, -1.4) (-1.2, -1) (-0.9, -0.6) (0.9, -0.6)} node (pt2) {};
        \fill[color=lightgp] plot[smooth cycle] coordinates {(2.1, -1.75) (2, -2.4) (0.6, -3.4) (0, -3.5) (-0.3, -3.3)(-0.3, -2.8) (0.1, -2.4) (1.2, -1.65) (1.7, -1.5) } node (pt3) {};
        \node[all, draw,fill=darkgray] (b0) at (0, 0) {$T_1$};
        \node[all, draw,fill=darkgray] (b1) at (-0.8, -1) {$T_2$};
        \node[all, draw,fill=white] (b2) at (0.8, -1) {$T_3$};
        \node[all, draw,fill=medgray] (b3) at (-1.5, -2) {$T_4$};
        \node[all, draw,fill=medgray2] (b4) at (-0.5, -2) {$T_5$};
        \node[all, draw,fill=white] (b5) at (0.9, -2.5) {$T_6$};
        \node[all, draw,fill=white] (b6) at (1.7, -2) {$T_7$};
        \node[all, draw,fill=lightgray] (b7) at (-2, -3) {$T_8$};
        \node[all, draw,fill=lightgray] (b8) at (0.2, -3) {$T_9$};
        \begin{scope}[line width=0.2ex]
        \path[->] (b1) edge[out=90, in=-110] node[sloped,above] {} (b0) ;
        \path[->] (b2) edge[out=100, in=-70] node[sloped,above] {} (b0) ;
        \path[->] (b3) edge[out=90, in=-110] node[sloped,above] {} (b1) ;
        \path[->] (b4) edge[out=90, in=-70] node[sloped,above] {} (b1) ;
        \path[->] (b5) edge[out=100, in=-110] node[sloped,above] {} (b2) ;
        \path[->] (b6) edge[out=120, in=-70] node[sloped,above] {} (b2) ;
        \path[->] (b7) edge[out=90, in=-90] node[sloped,above] {} (b3) ;
        \path[->] (b7) edge[out=0, in=-110] node[sloped,above] {} (b4) ;
        \path[->] (b8) edge[out=90, in=-70] node[sloped,above] {} (b4) ;
        \end{scope}
        \node[] (p1) at (3, 0.5) {\footnotesize$\mathcal{P}_{T_1}$} ;
        \node[] (p2) at (3, 0) {\footnotesize$\mathcal{P}_{T_2} = \mathcal{P}_{T_3}$};
        \node[] (p3) at (3, -0.5) {\footnotesize$\mathcal{P}_{T_9} = \mathcal{P}_{T_6} = \mathcal{P}_{T_7}$};
        \begin{scope}[color=darkgray,dashed,line width=0.2ex]
        \path[->] (p1) edge[out=180, in=0] node[sloped,above] {} (pt1) ;
        \path[->] (p2) edge[out=180, in=60] node[sloped,above] {} (pt2) ;
        \path[->] (p3) edge[out=-90, in=30] node[sloped,above] {} (pt3) ;
        \end{scope}

        \node[anchor=south, yshift=0.1cm] at (b0.north) {\footnotesize$\langle \mathbf{c_{T_1}}, d(T_1) \rangle = \langle\mathbf{1}, 6\rangle$};
        \node[anchor=east,xshift=-0.1cm] at (b1.west) {\footnotesize$\langle \mathbf{1}, 5 \rangle$};
        \node[anchor=west,xshift=0.1cm] at (b2.east) {\footnotesize$\langle \mathbf{0}, 0 \rangle$};
        \node[anchor=east,xshift=-0.1cm] at (b3.west) {\footnotesize$\langle \mathbf{1}, 2 \rangle$};
        \node[anchor=west,xshift=0.1cm] at (b4.east) {\footnotesize$\langle \mathbf{1}, 3 \rangle$};
        \node[anchor=north west,xshift=0.1cm] at (b5.south east) {\footnotesize$\langle \mathbf{0}, 0 \rangle$};
        \node[anchor=north west] at (b6.south east) {\footnotesize$\langle \mathbf{0}, 0 \rangle$};
        \node[anchor=north west] at (b8.south east) {\footnotesize$\langle \mathbf{1}, 1 \rangle$};
        \node[anchor=north,yshift=-0.2cm] at (b7.south) {\footnotesize$\langle \mathbf{1}, 1 \rangle$};
        
    \end{scope}
\end{tikzpicture}
    \captionof{figure}{Example of $\langle \textrm{chit}, \textrm{confidence}\rangle$ values.  Darker boxes indicate transactions with higher confidence values. At most one transaction in each shaded region will be accepted.}
    \label{fig:dag-cd}
\end{center}
\end{figure}

Figure~\ref{fig:dag-cd} illustrates a sample DAG built by Avalanche.
Similar to Snowball, sampling in Avalanche will create a positive feedback for the preference of a single transaction in its conflict set.
For example, because $T_2$ has larger confidence than $T_3$, its descendants are more likely collect chits in the future compared to $T_3$.

\tronly{Similar to Bitcoin, Avalanche leaves determining the acceptance point of a transaction to the application. An application supplies an \textsc{isAccepted} predicate that can take into account the value at risk in the transaction and the chances of a decision being reverted to determine when to decide.}{}

Committing a transaction can be performed through a \emph{safe early commitment}. For virtuous transactions, $T$ is accepted when it is the only transaction in its conflict set and has a confidence not less than threshold $\beta_1$.
As in Snowball, $T$ can also be accepted after a $\beta_2$ number of consecutive successful queries.
If a virtuous transaction fails to get accepted due to a problem
with parents, it could be accepted if reissued with different parents.
Figure~\ref{fig:gossipchain-ongen} shows how Avalanche entangles transactions. Because transactions that consume and generate the same UTXO do not conflict with each other, any transaction can be reissued with different parents.

\tronly{
Figure~\ref{fig:gossipchain-main} illustrates the protocol main loop
executed by each node.
In each iteration, the node attempts to select a transaction $T$ that has not
yet been queried.  If no such
transaction exists, the loop will stall until a new transaction is
added to $\mathcal{T}$.
It then selects $k$ peers and queries those peers.
If more than $\alpha$ of those peers return a positive response, the chit value is set to~1.
After that, it updates the preferred transaction of each conflict set
of the transactions in its ancestry.
Next, $T$ is added to the set $\mathcal{Q}$
so it will never be queried again by the node.
The code that selects additional peers if some of the $k$ peers are
unresponsive is omitted for simplicity.

Figure~\ref{fig:gossipchain-onquery} shows what happens when a node
receives a query for transaction $T$ from peer $j$.
First it adds $T$ to $\mathcal{T}$, unless it already has it.
Then it determines if $T$ is currently strongly preferred.
If so, the node returns a positive response to peer $j$.
Otherwise, it returns a negative response.
Notice that in the pseudocode, we assume when a node knows $T$, it also
recursively knows the entire ancestry of $T$. This can be achieved by
postponing the delivery of $T$ until its entire ancestry is recursively
fetched.
In practice, an additional gossip process that disseminates
transactions is used in parallel, but is not shown in pseudocode for simplicity.
}{}

\subsection{Multi-Input UTXO Transactions}
In addition to the DAG structure in Avalanche, an \emph{unspent transaction output} (UTXO)~\cite{nakamoto2008bitcoin} graph that captures
spending dependency is used to realize the ledger for the payment system. To
avoid ambiguity, we denote the transactions that encode the data for money
transfer \emph{transactions}, while we call the
transactions ($T \in \mathcal{T}$) in Avalanche's DAG \emph{vertices}.

We inherit the transaction and address mechanisms from Bitcoin. At their simplest, transactions consist of multiple inputs and outputs, with corresponding redeem scripts.
Addresses are identified by the hash of their public keys, and signatures are generated by corresponding private keys.
The full scripting language is used to ensure that a redeem script is authenticated to spend a UTXO\@.
UTXOs are fully consumed by a valid transaction, and may generate new UTXOs spendable by named recipients.
Multi-input transactions consume multiple UTXOs, and in Avalanche, may appear in multiple conflict sets.
To account for these correctly, we represent \emph{transaction-input} pairs (e.g. $\texttt{In}_{a1}$) as Avalanche vertices.
The conflict relation of transaction-input pairs is transitive because of each pair only spends one unspent output.
Then, we use the conjunction of \textsc{isAccepted} for all inputs of a transaction to ensure that no transaction will be accepted unless all its inputs are accepted (Figure~\ref{fig:cash-system-b}). In other words, a transaction is accepted only if all its transaction-input pairs are accepted in their respective Snowball conflict sets.
Following this idea, we finally implement the DAG of transaction-input pairs such that multiple
transactions can be batched together per query.



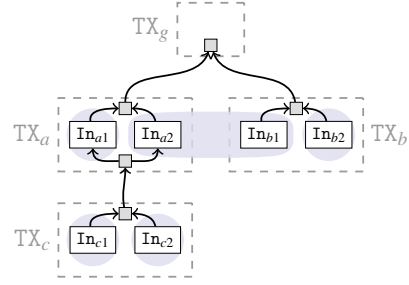
\begin{figure}[t]
\begin{center}
    \definecolor{lightgray}{HTML}{dddddd}
\definecolor{medgray}{HTML}{cccccc}
\definecolor{medgray2}{HTML}{bbbbbb}
\definecolor{darkgray}{HTML}{999999}
\definecolor{lightgp}{HTML}{ddddee}
\begin{tikzpicture}[x=1.12cm, scale=0.7, every node/.append style={transform shape}]
    \begin{scope}[all/.style={draw,fill=white,minimum height=0.5cm, minimum width=0.5cm},line width=0.08ex]
    \begin{scope}[color=darkgray,dashed,line width=0.2ex]
        \node[all, label={[anchor=east, font=\fontsize{14}{14}]left:$\texttt{TX}_g$},
        rectangle,minimum width=0.5in, minimum height=0.4in, line width=0.15ex] (a0) at (0, 0) {};
        \node[all, label={[anchor=east, font=\fontsize{14}{14}]left:$\texttt{TX}_a$}, rectangle,minimum width=1in, minimum height=0.55in, line width=0.15ex]
        (a1) at (-1.45, -2) {};
        \node[all, label={[anchor=west, font=\fontsize{14}{14}]right:$\texttt{TX}_b$}, rectangle,minimum width=1in, minimum height=0.55in, line width=0.15ex]
        (a2) at (1.45, -2) {};
        \node[all, label={[anchor=east, font=\fontsize{14}{14}]left:$\texttt{TX}_c$}, rectangle,minimum width=1in, minimum height=0.55in, line width=0.15ex]
        (a3) at (-1.45, -4) {};
    \end{scope}
    
        \node[fill=lightgp, circle,minimum width=1cm, opacity=0.8] (pt1) at (-2, -2) {};
        \node[fill=lightgp, circle,minimum width=1cm, opacity=0.8] (pt4) at (2, -2) {};
        \node[fill=lightgp, circle,minimum width=1cm, opacity=0.8] (pt5) at (-2, -4) {};
        \node[fill=lightgp, circle,minimum width=1cm, opacity=0.8] (pt6) at (-0.9, -4) {};
        \fill[color=lightgp, opacity=0.8] plot[smooth cycle] coordinates { (1.4, -2) (1.1, -2.4) (-1.1, -2.4) (-1.4, -2) (-1.1, -1.6) (1.1, -1.6)} node (pt2) {};
        
        \node[rectangle, fill=lightgray, draw, minimum width=1ex](ain) at (-1.45, -1.5) {};
        \node[rectangle, fill=lightgray, draw, minimum width=1ex](bin) at (1.45, -1.5) {};
        \node[rectangle, fill=lightgray, draw, minimum width=1ex](aout) at (-1.45, -2.5) {};
        \node[rectangle, fill=lightgray, draw, minimum width=1ex](cin) at (-1.45, -3.5) {};
        \node[rectangle, fill=lightgray, draw, minimum width=1ex](gout) at (0, -0.3) {};
        
        \node[all, draw] (b1) at (-2, -2) {$\texttt{In}_{a1}$};
        \node[all, draw] (b2) at (-0.9, -2) {$\texttt{In}_{a2}$};
        \node[all, draw] (b3) at (0.9, -2) {$\texttt{In}_{b1}$};
        \node[all, draw] (b4) at (2, -2) {$\texttt{In}_{b2}$};
        
        \node[all, draw] (b5) at (-2, -4) {$\texttt{In}_{c1}$};
        \node[all, draw] (b6) at (-0.9, -4) {$\texttt{In}_{c2}$};
        
        \begin{scope}[line width=0.15ex]
        \path[->] (ain) edge[out=90, in=-110] node[sloped,above] {} (gout) ;
        \path[->] (bin) edge[out=100, in=-70] node[sloped,above] {} (gout) ;
        \path[->] (cin) edge[out=100, in=-90] node[sloped,above] {} (aout) ;
        \path[->] (b1) edge[out=100, in=180] node[sloped,above] {} (ain);
        \path[->] (b2) edge[out=100, in=0] node[sloped,above] {} (ain);
        \path[->] (b3) edge[out=100, in=180] node[sloped,above] {} (bin);
        \path[->] (b4) edge[out=100, in=0] node[sloped,above] {} (bin);
        \path[->] (b5) edge[out=100, in=180] node[sloped,above] {} (cin);
        \path[->] (b6) edge[out=100, in=0] node[sloped,above] {} (cin);
        \path[->] (aout) edge[out=180, in=-90, looseness=1.3] node[sloped,above] {} (b1);
        \path[->] (aout) edge[out=0, in=-90, looseness=1.3] node[sloped,above] {} (b2);
        \end{scope}
        \begin{scope}[color=darkgray,dashed,line width=0.2ex]
        \end{scope}
    \end{scope}
\end{tikzpicture}
    \captionof{figure}{The underlying logical DAG structure used by Avalanche.
    The tiny squares with shades are dummy vertices which just help form the
    DAG topology for the purpose of clarity, and can be replaced by direct
    edges. The rounded gray regions are the conflict sets.}\label{fig:cash-system-b}
\end{center}
\end{figure}

\paragraph{Optimizations}
We implement some optimizations to help the system scale.
First, we use \emph{lazy updates} to the DAG, because the recursive definition for confidence may otherwise require a costly DAG traversal.
We maintain the current $d(T)$ value for each active vertex on the DAG, and update it only when a descendant vertex gets a chit.
Since the search path can be pruned at accepted vertices, the cost for an update is constant if the rejected vertices have a limited number of descendants and the undecided region of
the DAG stays at constant size.
Second, the conflict set could be large in practice, because a rogue client can generate a large volume of conflicting transactions.
Instead of keeping a container data structure for each conflict set, we create a mapping from each UTXO to the preferred transaction that stands as the representative for the entire conflict set.
This enables a node to quickly determine future conflicts, and the appropriate response to queries.
Finally, we speed up the query process by terminating early as soon as the $\alpha$ threshold is met, without waiting for $k$ responses.


\paragraph{DAG} Compared to Snowball, Avalanche introduces a DAG structure that entangles the fate of unrelated conflict sets, each of which is a single-decree instance.
This entanglement embodies a tension: attaching a virtuous transaction to undecided parents helps propel transactions towards a decision, while it puts transactions at risk of suffering liveness failures when parents turn out to be rogue.
We can resolve this tension and provide a liveness guarantee with the aid of two mechanisms.

First we adopt an adaptive parent selection strategy, where transactions are attached at the live edge of the DAG, and are retried with new parents closer to the genesis vertex. This procedure is guaranteed to terminate with uncontested, decided parents, ensuring that a transaction cannot suffer liveness failure due to contested, rogue transactions.
A secondary mechanism ensures that virtuous transactions with decided ancestry will receive sufficient chits. Correct nodes examine the DAG for virtuous transactions that lack sufficient progeny and emit no-op transactions to help increase their confidence.
With these two mechanisms in place, it is easy to see that, at worst, Avalanche will degenerate into separate instances of Snowball, and thus provide the same liveness guarantee for virtuous transactions.

Unlike other cryptocurrencies~\cite{IOTA} that use graph vertices
directly as votes, Avalanche only uses a DAG for the purpose of batching queries
in the underlying Snowball instances.
Because confidence is built by collected chits, and not by just the presence of
a vertex, simply flooding the network with vertices attached to the rejected
side of a subgraph will not subvert the protocol.

\subsection{Communication Complexity}
Let the DAG induced by Avalanche have an expected branching factor of $p$, corresponding to the width of the DAG, and determined by the parent selection algorithm.
Given the $\beta_1$ and $\beta_2$ decision threshold, a transaction that has just reached the point of decision will have an associated progeny $\mathcal{Y}$.
Let $m$ be the expected depth of $\mathcal{Y}$.
If we were to let the Avalanche network make progress and then freeze the DAG at a depth $y$,
then it will have roughly $py$ vertices/transactions, of which $p(y - m)$ are decided in expectation.
Only $pm$ recent transactions would lack the progeny required for a decision.
For each node, each query requires $k$ samples, and therefore the total message cost per transaction is in expectation $(pky) / (p(y - m)) = ky/(y-m)$.
Since $m$ is a constant determined by the undecided region of the DAG as the system constantly makes progress, message complexity per node is $O(k)$, while the total complexity is $O(kn)$.

\section{Evaluation}
\label{sec:evaluation}
\newcommand{\sysname}{Avalanche}


\subsection{Setup}
We conduct our experiments on Amazon EC2 by running from hundreds (125) to thousands (2000)
of virtual machine instances.  We use \texttt{c5.large} instances, each of
which simulates an individual node. AWS provides
bandwidth of up to 2 Gbps, though the {\sysname} protocol utilizes at most around 100 Mbps.

Our implementation supports two versions of transactions: one is the customized UTXO format,
while the other uses the code directly from Bitcoin 0.16. Both supported formats use secp256k1
crypto library from bitcoin and provide the same address format for wallets. All experiments
use the customized format except for the geo-replication, where results for both are given.

We simulate a constant flow of new transactions from users by creating
separate client processes, each of which maintains
separated wallets, generates transactions with new recipient addresses and
sends the requests to {\sysname} nodes.
We use several such client
processes to max out the capacity of our system.  The number of recipients
for each transaction is tuned to achieve average transaction sizes of
around 250 bytes (1--2 inputs/outputs per transaction on average and a stable
UTXO size), the current average transaction size of Bitcoin. To utilize the
network efficiently, we batch up to 40 transactions during a query, but
maintain confidence values at individual transaction granularity.

All reported metrics reflect end-to-end measurements taken from the perspective of all clients.
That is, clients examine the total number of confirmed
transactions per second for throughput, and, for each transaction,
subtract the initiation timestamp
from the confirmation timestamp for latency. Each throughput experiment is
repeated for 5 times and standard deviation is indicated in each figure.
As for security parameters, we pick $k = 10$, $\alpha = 0.8$, $\beta_1 = 11$, $\beta_2 = 150$, which yields an MTTF of \textasciitilde{}$10^{24}$ years.

\subsection{Throughput} 

We first measure the throughput of the system by saturating it with
transactions and examining the rate at which transactions are confirmed in the
steady state.  For this experiment, we first run {\sysname} on 125 nodes
with 10 client processes, each of which maintains 400 outstanding transactions at
any given time.

As shown by the first group of bars in Figure~\ref{fig:eval-thr}, the system achieves
6851 transactions per second (tps) for a batch size of 20 and above 7002 tps for a batch size of 40.
Our system is saturated by a small batch size comparing to other blockchains with known performance:
Bitcoin batches several thousands of transactions per block, Algorand~\cite{GiladHMVZ17} uses 2--10 Mbyte blocks, i.e., 8.4--41.9K tx/batch and Conflux~\cite{confluxLLXLC18} uses 4 Mbyte blocks, i.e., 16.8K tx/batch. These systems are relatively slow in making a single decision, and thus require a very large batch (block) size for better performance. Achieving high throughput with small batch size implies low latency, as we will show later.

\begin{figure}[h]
\includegraphics[width=\linewidth]{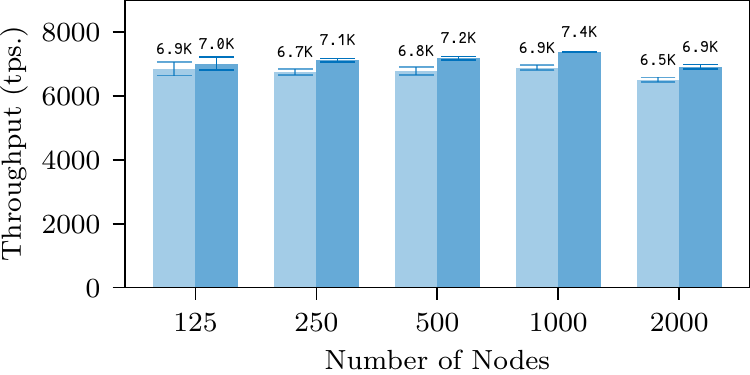}
\captionof{figure}{Throughput vs. network size. Each pair of bars is produced with batch size of 20 and 40, from left to right.}
\label{fig:eval-thr}
\end{figure}

\begin{figure}
\includegraphics[width=\linewidth]{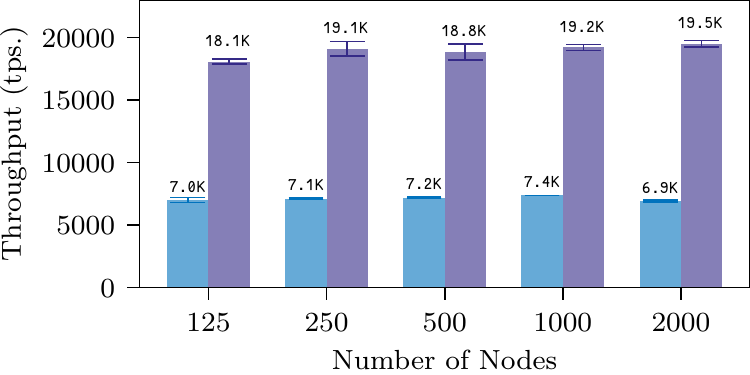}
\captionof{figure}{Throughput for batch size of 40, with (left) and without (right) signature verification.}
\label{fig:eval-thr-raw}
\end{figure}

\subsection{Scalability}

To examine how the system scales in terms of the number of nodes
participating in {\sysname} consensus, we run experiments with identical settings
and vary the number of nodes from 125 up to 2000.

Figure~\ref{fig:eval-thr} shows that overall throughput degrades about $1.34\%$ to 6909 tps when the network grows by a factor of 16 to $n = 2000$.
This degradation is minor compared to the growth of the network size.
Note that the x-axis is logarithmic.

Avalanche acquires its scalability from three sources: first,
maintaining a partial order that captures only the spending relations
allows for more concurrency than a classical BFT replicated
log that linearizes all transactions; second, the lack of a leader naturally avoids bottlenecks;
finally, the number of messages each node has to handle per decision is $O(k)$ and does not grow as the network scales up.


\subsection{Cryptography Bottleneck}

We next examine where bottlenecks lie in our current implementation.
The purple bar on the right of each group in Figure~\ref{fig:eval-thr-raw} shows the throughput of Avalanche with signature verification
disabled. Throughputs get approximately 2.6x higher, compared to the blue bar on the left.
This reveals that cryptographic verification overhead is the current bottleneck of our system implementation.
This bottleneck can be addressed by offloading
transaction verification to a GPU\@. Even without such optimization,
7K tps is far in excess of extant blockchains.

\subsection{Latency}

\begin{figure}
\includegraphics[width=\linewidth]{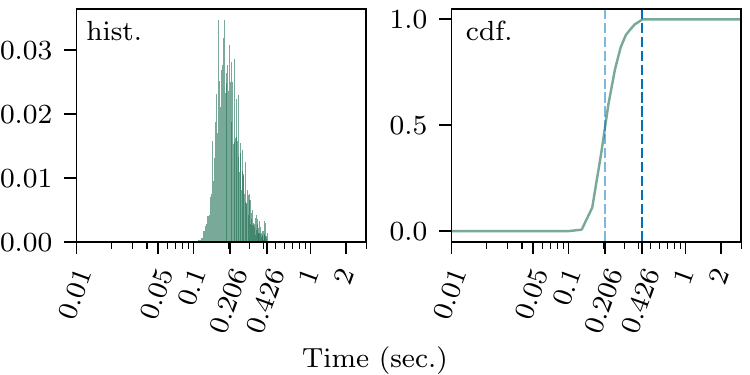}
\captionof{figure}{Transaction latency distribution for $n = 2000$. The x-axis
    is the transaction latency in log-scaled seconds, while the
    y-axis is the portion of transactions that fall into the confirmation
    time (normalized to $1$).  Histogram of all transaction latencies for a client is shown on the left with $100$ bins,
    while its CDF is on the right.}
\label{fig:eval-lat1}
\end{figure}

The latency of a transaction is the time spent from the moment of its
submission until it is confirmed as accepted.
Figure~\ref{fig:eval-lat1} tallies
the latency distribution histogram using the same setup as for the throughput
measurements with 2000 nodes. The x-axis is the time in seconds while the y-axis
is the portion of transactions that are finalized within the corresponding
time period. This figure also outlines the Cumulative Distribution Function (CDF)
by accumulating the number of finalized transactions over time.

This experiment shows that most transactions are confirmed within approximately 0.3 seconds.
The most common latencies are around 206 ms and variance is low,
indicating that nodes converge on the final value as a group around the same time.
The second vertical line shows the maximum latency we observe, which is around 0.4 seconds.

Figure~\ref{fig:eval-lat2} shows
transaction latencies for different numbers of nodes.
The horizontal edges of boxes represent minimum,
first quartile, median, third quartile and maximum latency respectively, from
bottom to top.  Crucially, the experimental data show that median latency is more-or-less
independent of network size.

\begin{figure}[h]
\includegraphics[width=\linewidth]{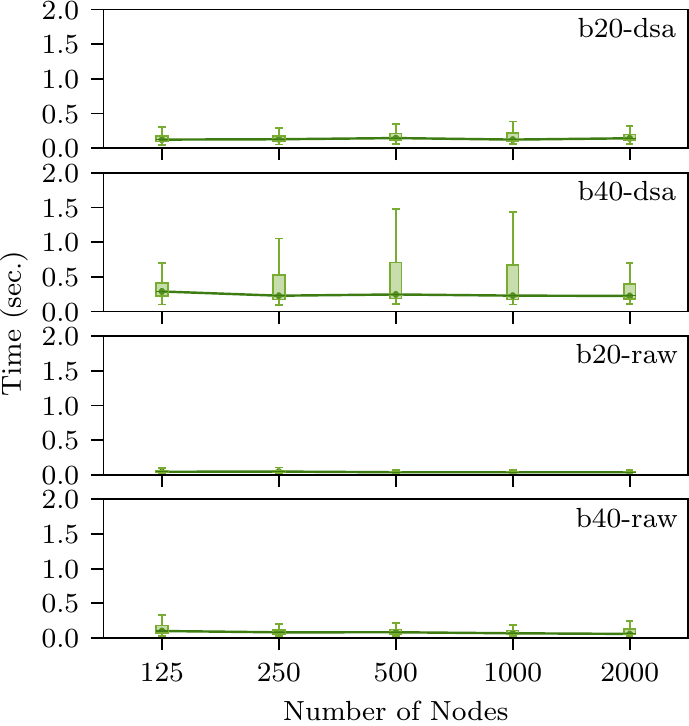}
\captionof{figure}{Transaction latency vs. network size. ``b'' indicates batch size and ``raw'' is the run without signature verification.}
\label{fig:eval-lat2}
\end{figure}

\subsection{Misbehaving Clients}
\label{sec:evaluation-misbehaving}
\begin{figure}
\includegraphics[width=\linewidth]{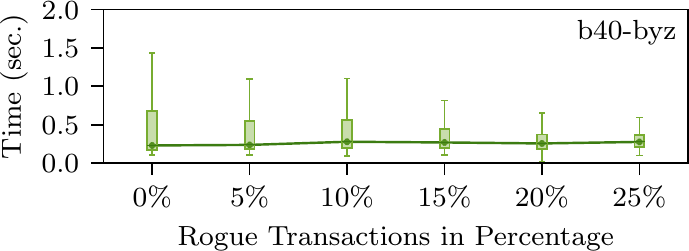}
\captionof{figure}{Latency vs. ratio of rogue transactions.}
\label{fig:eval-lat3}
\end{figure}

We next examine how rogue transactions issued by misbehaving clients that
double spend unspent outputs can affect latency for virtuous
transactions created by honest clients. We adopt a strategy to
simulate misbehaving clients where a fraction (from $0\%$ to $25\%$) of the
pending transactions conflict with some existing ones.
The client processes
achieve this by designating some double spending transaction flows among all
simulated pending transactions and sending the conflicting transactions to
different nodes. We use the same setup with $n = 1000$ as in the previous
experiments, and only measure throughput and latency of confirmed transactions.

{\sysname}'s latency is only slightly affected by misbehaving clients, as shown
in Figure~\ref{fig:eval-lat3}. Surprisingly, maximum latencies drop slightly when
the percentage of rogue transactions increases.  This behavior occurs
because, with the introduction of rogue transactions, the overall
\emph{effective} throughput is reduced and thus alleviates system load.
This is confirmed by
Figure~\ref{fig:eval-thr2}, which shows that throughput (of virtuous transactions) decreases with the ratio of rogue transactions.
Further, the reduction in throughput appears proportional to the number of misbehaving clients,
that is, there is no leverage provided to the attackers.

\begin{figure}
\includegraphics[width=\linewidth]{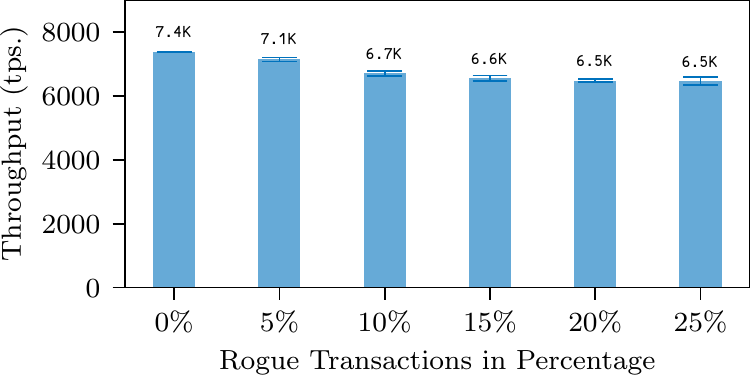}
\captionof{figure}{Throughput vs. ratio of rogue transactions.}
\label{fig:eval-thr2}
\end{figure}

\subsection{Geo-replication}

\begin{figure}
\includegraphics[width=\linewidth]{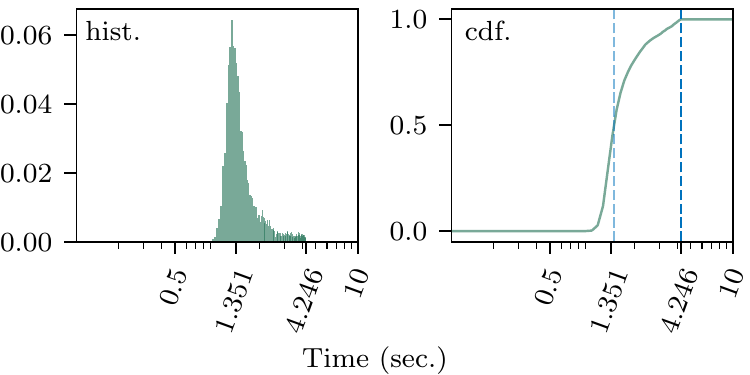}
\captionof{figure}{Latency histogram/CDF for $n = 2000$ in 20 cities.}
\label{fig:eval-lat4}
\end{figure}

Next experiment shows the system in an emulated geo-replicated scenario, patterned
after the same scenario in prior work~\cite{GiladHMVZ17}.
We selected 20 major cities that appear to be near substantial numbers of reachable
Bitcoin nodes, according to~\cite{bitnodes2018}. The cities cover North America,
Europe, West Asia, East Asia, Oceania, and also cover the top 10 countries with the
highest number of reachable nodes. We use the latency and jitter matrix crawled
from~\cite{wondernetworkping2018} and emulate network packet latency in the Linux kernel using
\texttt{tc} and \texttt{netem}. 2000 nodes are distributed evenly to each
city, with no additional network latency emulated between nodes within the same city.
Like Algorand's evaluation, we also cap our bandwidth per process to 20 Mbps
to simulate internet-scale settings where there are many commodity network links.
We assign a client process to each city, maintaining 400 outstanding transactions per city
at any moment.

In this scenario, Avalanche achieves an average throughput of 3401 tps, with a standard deviation of 39 tps.
As shown in Figure~\ref{fig:eval-lat4}, the median transaction latency
is 1.35 seconds, with a maximum latency of 4.25 seconds. We also support native Bitcoin code
for transactions; in this case, the throughput is 3530 tps, with $\sigma = 92$ tps.

\subsection{Comparison to Other Systems}
Though there are seemingly abundant blockchain or cryptocurrency protocols,
most of them only present a sketch of their protocols and do not offer practical implementation or evaluation results.
Moreover, among those who do provide results, most are not evaluated in realistic, large-scale (hundreds to thousands of full nodes participating in consensus) settings.

Therefore, we choose Algorand and Conflux for our comparison. Algorand, Conflux, and Avalanche are all fundamentally different in their
design. Algorand's committee-scale consensus algorithm is quorum-based Byzantine agreement, and Conflux extends Nakamoto consensus by a DAG structure to facilitate higher throughput, while Avalanche belongs to a new protocol family based on metastability. Additionally, we use Bitcoin~\cite{nakamoto2008bitcoin} as a baseline.

Both Algorand and Avalanche evaluations use a decision network of size 2000 on EC2.
Our evaluation picked shared \texttt{c5.large} instances, while Algorand used \texttt{m4.2xlarge}\@.
These two platforms are very similar except for a slight CPU clock speed edge for \texttt{c5.large}, which goes largely unused because our process only consumes $30\%$ in these experiments.
The security parameters chosen in our experiments guarantee a safety violation
probability below $10^{-9}$ in the presence of $20\%$ Byzantine nodes, while
Algorand's evaluation guarantees a violation probability below $5 \times 10^{-9}$ with $20\%$ Byzantine nodes.

Neither Algorand nor Conflux evaluations take into account the overhead of cryptographic verification.
Their evaluations use blocks that carry megabytes of dummy data and present the throughput in MB/hour or GB/hour unit. So we use the average size of a Bitcoin transaction, 250 bytes, to derive their throughputs. In contrast, our experiments carry real transactions and fully take all cryptographic overhead into account.

The throughput is 3-7 tps for Bitcoin,
874 tps for Algorand (with 10 Mbyte blocks),
3355 tps for Conflux (in the paper it claims 3.84x Algorand's throughput under the same settings).

In contrast, {\sysname} achieves over 3400 tps
consistently on up to 2000 nodes without committee or proof-of-work.
As for latency, a transaction is confirmed after 10--60
minutes in Bitcoin, around 50 seconds in Algorand,
7.6--13.8 minutes in Conflux, and 1.35 seconds in {\sysname}.

Avalanche performs much better than Algorand in both throughput and latency
because Algorand uses a verifiable random function to elect committees, and
maintains a totally-ordered log while {\sysname} establishes only a partial
order. Although Algorand's leadership is anonymous and changes continuously,
it is still leader-based which could be the bottleneck for scalability, while
{\sysname} is leader-less.

Avalanche has similar throughput to Conflux, but its latency is 337--613x better.
Conflux also uses a DAG structure to amortize the cost for consensus and increase the throughput,
however, it is still rooted in Nakamoto consensus (PoW), making it unable to
have instant confirmation compared to Avalanche.

In a blockchain system, one can usually improve throughput at the cost of latency through batching. The real bottleneck of the performance is the number of decisions the system can make per second, and this is fundamentally limited by either Byzantine Agreement ($\mathrm{BA}^{*}$) in Algorand and Nakamoto consensus in Conflux.

\section{Related Work}
\label{sec:related-work}
\label{section:background}
\tronly{
Bitcoin~\cite{nakamoto2008bitcoin} is a cryptocurrency that uses a blockchain based on
proof-of-work (PoW) to maintain a ledger of UTXO transactions. 
While techniques based on proof-of-work~\cite{DworkN92, aspnes2005exposing}, and even cryptocurrencies with minting based on proof-of-work~\cite{vishnumurthy2003karma,rivest1997payword}, have been explored before, Bitcoin was the first to incorporate PoW into its consensus process.
Unlike more traditional BFT protocols, Bitcoin has a probabilistic safety guarantee
and assumes honest majority computational power rather than a known
membership, which in turn has enabled an internet-scale permissionless protocol. While permissionless and resilient to adversaries,
Bitcoin suffers from low throughput (\textasciitilde{}3 tps) and
high latency (\textasciitilde{}5.6 hours for a network with 20\% Byzantine presence and $2^{-32}$ security guarantee).  Furthermore, PoW requires a substantial amount of computational power that is consumed only for the purpose of maintaining
safety.

Countless cryptocurrencies use PoW~\cite{DworkN92, aspnes2005exposing} to maintain a distributed ledger. 
Like Bitcoin, they suffer from inherent scalability bottlenecks. 
Several proposals for protocols exist that try to
better utilize the effort made by PoW.
Bitcoin-NG~\cite{EyalGSR16} and the permissionless version of
Thunderella~\cite{PassS18} use Nakamoto-like consensus to elect a leader that
dictates writing of the replicated log for a relatively long time so as to
provide higher throughput. Moreover, Thunderella provides an
optimistic bound that, with 3/4 honest computational power and an honest
elected leader, allows transactions to be confirmed rapidly.
ByzCoin~\cite{Kokoris-KogiasJ16} periodically selects a small set of
participants and then runs a PBFT-like protocol within the selected nodes.

Protocols based on Byzantine agreement~\cite{PeaseSL80, LamportSP82} typically make use of
quorums and require precise knowledge of membership.
PBFT~\cite{castro1999practical, CL02}, a well-known representative, requires a quadratic number of message exchanges in order to
reach agreement. 
The Q/U protocol~\cite{abd2005fault} and HQ replication~\cite{cowling2006hq} use a quorum-based approach to optimize for contention-free cases of operation to achieve consensus in only a single round of communication. However, although these protocols improve on performance, they degrade very poorly under contention. Zyzzyva~\cite{KotlaADCW09} couples BFT with speculative execution to improve the failure-free operation case.
Past work in permissioned BFT systems typically requires at least $3f+1$ replicas. CheapBFT~\cite{kapitza2012cheapbft} leverages trusted hardware components to construct a protocol that uses $f+1$ replicas.

Other work attempts to introduce new protocols under redefinitions and relaxations of the BFT model. 
Large-scale BFT~\cite{rodrigues2007large} modifies PBFT to allow for arbitrary choice of number of replicas and failure threshold, providing a probabilistic guarantee of liveness for some failure ratio but protecting safety with high probability. 
In another form of relaxation. Zeno~\cite{singh2009zeno} introduces a BFT state machine replication protocol that trades consistency for high availability. More specifically, Zeno guarantees eventual consistency rather than linearizability, meaning that participants can be inconsistent but eventually agree once the network stabilizes. By providing an even weaker consistency guarantee, namely fork-join-causal consistency, Depot~\cite{mahajan2011depot} describes a protocol that guarantees safety under $2f+1$ replicas. 

NOW~\cite{guerraoui2013highly} uses sub-quorums to drive smaller instances of consensus. The insight of this paper is that small, logarithmic-sized quorums can be extracted from a potentially large set of nodes in the network, allowing smaller instances of consensus protocols to be run in parallel. 

Snow White~\cite{cryptoeprint:2016:919} and
Ouroboros~\cite{KiayiasRDO17} are some of the earliest provably secure PoS
protocols.  Ouroboros uses a secure multiparty coin-flipping protocol to
produce randomness for leader election. The follow-up protocol, Ouroboros
Praos~\cite{DavidGKR18} provides safety in the presence of fully adaptive
adversaries.
HoneyBadger~\cite{MillerXCSS16} provides good liveness in a network with heterogeneous latencies. 

Tendermint~\cite{buchman2016tendermint, 1807.04938} rotates the leader for each block
and has been demonstrated with as many as 64 nodes. Ripple~\cite{schwartz2014ripple} has low latency by utilizing collectively-trusted
sub-networks in a large network. The Ripple company provides a
slow-changing default list of trusted nodes, which renders the system essentially centralized.
HotStuff~\cite{hotstuff,hotstuffpodc} improves the communication cost from quadratic to linear and significantly simplifies the protocol specification, although the leader bottleneck still persists. Facebook uses HotStuff as the core consensus for its Libra project.
In the synchronous setting, inspired by HotStuff, Sync HotStuff~\cite{synchotstuff} achieves consensus in $2\Delta$ time with quadratic cost and unlike other lock-steped synchronous protocols, it operates as fast as network propagates.
Stellar~\cite{mazieres2015stellar} uses Federated Byzantine Agreement in which \emph{quorum slices}
enable heterogeneous trust for different nodes.  Safety is guaranteed when
transactions can be transitively connected by trusted quorum slices.
Algorand~\cite{GiladHMVZ17} uses a verifiable random function to select a
committee of nodes that participate in a novel Byzantine consensus
protocol.

Some protocols use a Directed Acyclic Graph (DAG) structure instead of a linear chain to achieve
consensus~\cite{SompolinskyZ15,SompolinskyLZ16,SompolinskyZ18,BentovHMN17,baird2016hashgraph}.
Instead of choosing the longest chain as in Bitcoin,
GHOST~\cite{SompolinskyZ15} uses a more efficient chain selection rule that
allows transactions not on the main chain to be taken into consideration, increasing efficiency.
SPECTRE~\cite{SompolinskyLZ16} uses transactions on
the DAG to vote recursively with PoW to achieve consensus, followed up by
PHANTOM~\cite{SompolinskyZ18} that achieves a linear order among all blocks.
Like PHANTOM, Conflux also finalizes a linear order of transactions by PoW
in a DAG structure, with better resistance to liveness attack~\cite{confluxLLXLC18}.
Avalanche
is different in that the voting result is a one-time chit that is determined by
a query without PoW, while the votes in PHANTOM or Conflux are purely determined by PoW in transaction structure.
Similar to Thunderella, Meshcash~\cite{BentovHMN17} combines a slow PoW-based protocol with a fast consensus protocol that allows a high block rate regardless of network latency, offering fast confirmation time.
Hashgraph~\cite{baird2016hashgraph} is a leader-less protocol that builds a DAG via randomized gossip.
It requires full membership knowledge at all times, and, similar to the Ben-Or~\cite{ben1983another}, it requires exponential messages~\cite{aspnes2003randomized,CachinV17} in expectation.
}{
Several proposals for protocols exist that try to
better utilize the effort made by PoW.
Bitcoin-NG~\cite{EyalGSR16} and the permissionless version of
Thunderella~\cite{PassS18} use Nakamoto consensus to elect a leader that
dictates writing of the replicated log for a relatively long time so as to
provide higher throughput. Moreover, Thunderella provides an
optimistic bound that, with 3/4 honest computational power and an honest
elected leader, allows transactions to be confirmed rapidly.
ByzCoin~\cite{Kokoris-KogiasJ16} periodically selects a small set of
participants and then runs a PBFT-like protocol within the selected nodes. 

Protocols based on Byzantine agreement~\cite{PeaseSL80, LamportSP82} typically make use of
quorums and require precise knowledge of membership.
PBFT~\cite{castro1999practical, CL02}, a well-known representative, requires a quadratic number of message exchanges in order to
reach agreement. Some variants are able to scale to dozens of nodes~\cite{ClementWADM09, KotlaADCW09}.
Snow White~\cite{cryptoeprint:2016:919} and
Ouroboros~\cite{KiayiasRDO17} are some of the earliest provably secure PoS
protocols.  Ouroboros uses a secure multiparty coin-flipping protocol to
produce randomness for leader election. The follow-up protocol, Ouroboros
Praos~\cite{DavidGKR18} provides safety in the presence of fully adaptive
adversaries.
HoneyBadger~\cite{MillerXCSS16} provides good liveness in a network with heterogeneous latencies. 
Tendermint~\cite{buchman2016tendermint, 1807.04938} rotates the leader for each block
and has been demonstrated with as many as 64 nodes. Ripple~\cite{schwartz2014ripple} has low latency by utilizing collectively-trusted
sub-networks in a large network. The Ripple company provides a
slow-changing default list of trusted nodes, which renders the system essentially centralized.
HotStuff~\cite{hotstuff,hotstuffpodc} improves the communication cost from quadratic to linear and significantly simplifies the protocol specification, although the leader bottleneck still persists. Facebook uses HotStuff as the core consensus~\cite{librabft} for the Libra project.
In the synchronous setting, inspired by HotStuff, Sync HotStuff~\cite{synchotstuff} achieves consensus in $2\Delta$ time with quadratic cost and unlike other lock-steped synchronous protocols, it operates as fast as network propagates.
Stellar~\cite{mazieres2015stellar} uses Federated Byzantine Agreement in which \emph{quorum slices}
enable heterogeneous trust for different nodes.  Safety is guaranteed when
transactions can be transitively connected by trusted quorum slices.

Some protocols use a Directed Acyclic Graph (DAG) structure instead of a linear chain to achieve
consensus~\cite{SompolinskyZ15,SompolinskyLZ16,SompolinskyZ18,BentovHMN17,baird2016hashgraph}.
Instead of choosing the longest chain as in Bitcoin,
GHOST~\cite{SompolinskyZ15} uses a more efficient chain selection rule that
allows transactions not on the main chain to be taken into consideration, increasing efficiency.
SPECTRE~\cite{SompolinskyLZ16} uses transactions on
the DAG to vote recursively with PoW to achieve consensus, followed up by
PHANTOM~\cite{SompolinskyZ18} that achieves a linear order among all blocks.
Avalanche
is different in that the voting result is a one-time chit that is determined by
a query without PoW, while the votes in PHANTOM are purely determined by transaction structure.
Similar to Thunderella, Meshcash~\cite{BentovHMN17} combines a slow PoW-based protocol with a fast consensus protocol that allows a high block rate regardless of network latency, offering fast confirmation time.
Hashgraph~\cite{baird2016hashgraph} is a leader-less protocol that builds a DAG via randomized gossip. 
It requires full membership knowledge at all times, and, similar to the Ben-Or~\cite{ben1983another}, it requires exponential messages~\cite{aspnes2003randomized,CachinV17} in expectation.
}

%


\section{Conclusion}
\label{sec:conclusions}
This paper introduced a novel family of consensus protocols, coupled with the appropriate mathematical tools for analyzing them.
\tronly{These protocols are highly efficient and robust, combining the best features of classical and Nakamoto consensus.}
They scale well, achieve high throughput and quick finality, work without precise membership knowledge, and degrade gracefully under catastrophic adversarial attacks.

There is much work to do to improve this line of research. \tronly{
One such improvement could be the introduction of an adversarial network scheduler.
Another}{One} improvement would be to characterize the system's guarantees under an adversary whose powers are realistically limited, whereupon performance would improve even further. \tronly{Finally, more}{More} sophisticated initialization mechanisms would bear fruitful in improving liveness of multi-value consensus.
Overall, we hope that the protocols and analysis techniques presented here add to the arsenal of the distributed system developers and provide a foundation for new lightweight and scalable mechanisms.


\bibliographystyle{acm}
\bibliography{paper}
\begin{appendices}
\section{Analysis}
\label{sec:full-analysis}
In this appendix, we provide an analysis of Slush, Snowflake and Snowball.

\subsection{Preliminaries}
We assume the network model as discussed in Section~\ref{sec:model_and_goals}. We let $\mathtt{R}$ (``red'') and $\mathtt{B}$ (``blue'') represent two generic conflicting choices.
Without loss of generality, we focus our attention on counts of $\mathtt{B}$, i.e.\ the total number of nodes that prefer blue.

\paragraph{Hypergeometric Distribution} Each network query of $k$ peers corresponds to a sample without replacement out of a network of $n$ nodes, also referred to as a hypergeometric sample.
We let the random variable $\mathcal{H}(\mathcal{N}, x, k) \rightarrow \{0, \dots, k\}$ denote the resulting counts of $\mathtt{B}$ in the sample (unless otherwise stated), where $x$ is the total count of $\mathtt{B}$ in the population. The probability that the query achieves the required threshold of $\alpha$ or more votes is given by:
\begin{equation}
\small
P(\mathcal{H}(\mathcal{N}, x, k) \geq \alpha) = \left.\sum_{j = \alpha}^{k} {x \choose j} {n - x \choose k - j} \middle/ {n \choose k}\right.
\label{eq:hypergeometric}
\end{equation}
For ease of notation, we overload $\mathcal{H}(*)$ by implicitly referring to $P(\mathcal{H}(\mathcal{N}, x, k) \geq \alpha)$ as $\mathcal{H}(\mathcal{N}, x, k, \alpha)$. 

\paragraph{Tail Bounds On Hypergeometric Distribution} We can reduce some of the complexity in Equation~\ref{eq:hypergeometric} by introducing a bound on the hypergeometric distribution induced by $\mathcal{H}^k_{\mathcal{N},x}$.
Let $p=x/n$ be the ratio of support for $\mathtt{B}$ in the population.
The expectation of $\mathcal{H}(\mathcal{N}, x, k)$ is exactly $kp$.
Then, the probability that $\mathcal{H}(\mathcal{N}, x, k)$ will deviate from the mean by more than some small constant $\psi$ is given by the Hoeffding tail bound \cite{hoeffding1963probability}, as follows,
\begin{equation}
\begin{split}
    P(\mathcal{H}(\mathcal{C}, x, k) \leq (p-\psi)k) &\leq e^{-k\mathcal{D}(p-\psi, p)}\\
    &\leq e^{-2(p-\psi)^2k}
\end{split}
\end{equation}
where $\mathcal{D}(p-\psi, p)$ is the Kullback-Leibler divergence, measured as
\begin{equation}
\begin{split}
    \mathcal{D}(a, b) &= a \log \frac{a}{b} + (1 - a) \log \frac{1 - a}{1 - b}
\end{split}
\end{equation}

\paragraph{Concentration of Sub-Martingales}
Let $\{X_{\{t \geq 0\}}\}$ be a sub-martingale and $|X_t - X_{t-1}| < c_t$ almost surely. Then, for all positive reals $\psi$ and all positive integers $t$, 
\begin{equation}
\begin{split}
P(X_t \geq X_0 + \psi) \leq e^{-\psi^2 / 2\sum_{i = 1}^{t} c_t^2}
\end{split}
\label{eq:submartingale}
\end{equation}

\subsection{Slush}
Slush operates in a non-Byzantine setting; that is, $f = 0$, $c = n$.
In this section, we will characterize the irreversibility properties of Slush (which appear in Snowflake and Snowball), as well as the precise converge rate distribution. The distribution of of both safety and liveness of Slush translate well to the Byzantine setting.


The procedural version of Slush in Figure~\ref{fig:slush-loop} made use of a parameter $m$, the number of rounds that a node executes Slush queries. 
What we ultimately want to extract is the total number of rounds $\phi$ that the scheduler will need to execute in order to guarantee that the entire network is the same color, whp.

We analyze the system mainly using a continuous time process. Let $\{X_{\{t \geq 0\}}\}$ be a CTMC.
The state space $\mathcal{S}$ of the stochastic process is a condensed version of the full configuration space, where each state $\{0, \dots, n\}$ represents the total number of blue nodes in the system. 

Let $\mathcal{F}_{X_s}$ be the filtration, or the history pertaining to the process, up to time $s$. This process is Markovian and time-homogeneous, conforming to 
\[
    P\{X_t = j | \mathcal{F}_{X_s}\} = P\{X_t = j | X_s\} = P\{X_t = j | X_0\}    
\]
Throughout the paper, we use $Q \equiv (q_{ij}, i, j \in \mathcal{S})$ notation to refer to the infinitesimal generator of the process, where death ($i \rightarrow i-1$) and birth ($i \rightarrow i+1$) rates of configuration transitions are denoted via $\mu_i$ and $\lambda_i$ ($\lambda_ i$ is distinct from the clock parameter $\lambda$, and will be clear from context). These rates are 
\[
    \begin{cases}
        \mu_i = i\ \mathcal{H}(\mathcal{N}, c-i, k, \alpha), & \text{for } i \rightarrow i - 1 \\
        \lambda_i = (c-i)\ \mathcal{H}(\mathcal{N}, i, k, \alpha), & \text{for } j \rightarrow i + 1 \\
    \end{cases}
\]
for $1\leq i\leq c-1$, and where $i = 0$ and $i = c$ are absorbing. Let $p_{ij}(t)$ refer to the probability of transitioning from state $i$ to $j$ at time $t$. 
We always assume that 
\[
    p_{ij}(t) = 
    \begin{cases}
      \lambda_it + o(t), & \text{for } j = i + 1 \\
      \mu_it + o(t), & \text{for } j = i - 1 \\
      1 - (\lambda_i + \mu_i)t + o(t), & \text{for } j = i \\
      o(t), & \text{otherwise }\\
    \end{cases}
\]
where all $o(t)$ are uniform in $i$. 

\paragraph{Irreversibility}
In Section~\ref{sec:analysis}, we discussed the loose Chvatal bound which provided intuitive understanding into the strong irreversibility dynamics of our core subsampling mechanism. In particular, once the network drifts to some majority value, it tends to revert back with only an exponentially small probability. We compute the closed-form expression for reversibility, and show that it is exponentially small.
\begin{theorem}
\label{theorem:slush_prob_convergence_minority}
Let $\xi_\delta$ be the probability of absorption into the all-red state ($s_0$), starting from a drift of $\delta$ (i.e. $\delta$ drift away from $n/2$). Then, assuming $\delta > 1$, 
\begin{equation}
\xi_\delta = 1 - \ddfrac{\sum_{l = 1}^{\delta} \prod_{i = 1}^{l-1} \mu_i^2 \prod_{j = l}^{n-l}\lambda_j}{2\sum_{l = 1}^{n/2}\prod_{i=1}^{l-1}\mu_i^2\prod_{j=l}^{n-l}\mu_j}
\end{equation}
and
\begin{equation}
\begin{split}
\ddfrac{\xi_{\delta} - \xi_{\delta+1}}{\xi_{\delta+1} - \xi_{\delta+2}} &= \mathcal{u}_{\delta+1} = \ddfrac{\lambda_{\delta+1}}{\mu_{\delta+1}} \\
&\approx \ddfrac{n-\delta-1 \sum_{j = \alpha}^{k}\ddfrac{(n-\delta-1)^k (\delta+1)^{k-j}}{n^{2k - j}}}{\delta+1 \sum_{j = \alpha}^{k}\ddfrac{(\delta+1)^k (n-\delta-1)^{k-j}}{n^{2k - j}}}
\end{split}
\end{equation}
where from now on we refer to $\mathcal{u}_{\delta+1}$ as the drift of the process. 
\end{theorem}

\begin{proof}
Our results are derived based on constructions from Tan~\cite{tan1976absorption}. We construct a sub-matrix of $Q$, denoted $B$, as shown in Figure~\ref{fig:matrixB}.
\begin{figure*}
\[B = 
\begin{bmatrix}
    -(\lambda_1 + \mu_1) & \lambda_1 & 0 & \cdots & \cdots & 0 \\
    \mu_2 & -(\lambda_2 + \mu_2) & \lambda_2 & 0 & \cdots & 0\\
    0 & \mu_3 & -(\lambda_3 + \mu_3) & \lambda_3 & \cdots & 0\\
    \vdots & \vdots & \ddots & \ddots & \ddots & \vdots\\
    \vdots & \vdots & \mu_{n-3} & -(\lambda_{n-2} + \mu_{n-2}) & \lambda_{n-3} & 0\\
    \vdots & \dots & 0 & \mu_{n-1} & -(\lambda_{n-2} + \mu_{n-2}) & \lambda_{n-2}\\
    0 & \dots & 0 & 0 & \mu_{n-1} & -(\lambda_{n-1} + \mu_{n-1})
\end{bmatrix}
\]
\caption{Matrix $B$.}
\label{fig:matrixB}
\end{figure*}
Let $W_1'$ = $(\mu_1, 0, \dots, 0)$, $W_{n-1}'$ = $(0, \dots, 0, \lambda_{n-1})$. Then, we can express $Q$ as 
\[
    Q =
    \begin{bmatrix}
        0 & \dots & 0\\
        W_1 & B & W_{n-1}\\
        0 & \dots & 0
    \end{bmatrix}
\]
As a reminder, the stationary distribution can be found via $\lim_{t \rightarrow \infty} P(t) = e^{Qt}$, where we have
\[
    e^{Qt} = \sum_{i = 0}^{\infty} \frac{t^i}{i!} Q^i = \sum_{i = 0}^{\infty} \frac{t^i}{i!}
    \begin{bmatrix}
        0 & \dots & 0\\
        B^{i-1}W_1 & B^i & B^{i-1}W_{n-1}\\
        0 & \dots & 0
    \end{bmatrix}
\]
As Tan (eq. 2.3) shows, we have
\[
    \xi(t) = B^{-1}\left[\sum_{i = 0}^{\infty} B^i  - \mathbb{I}_{n-1} \right] W_1
\]
Since we want the ultimate probabilities, we have that 
\[
    \xi = \lim_{t \rightarrow \infty} \xi(t) = -B^{-1}W_1
\]
We can explicitly compute $\xi_\delta$ in terms of our rates $\mu_i$ and $\lambda_i$, getting 
\[
    \xi_\delta = \ddfrac{\sum_{l = 1}^{n-\delta}\prod_{i = 1}^{n-l}\mu_i \prod_{j = n-l+1}^{n-1}\lambda_j}{\sum_{l = 1}^{n}\prod_{i = 1}^{n-l}\mu_i \prod_{j = n-l+1}^{n-1}\lambda_j}
\]
However, we note that $u_{i} = \lambda_{n-i}$. Algebraic manipulation from this observation leads to the two equations in the theorem. This expression is strictly lower than the Chvatal bounds used in Section~\ref{sec:analysis}.
\end{proof}

Using the construction for the absorption (and (ir)reversibility) probabilities as discussed previously, a natural follow up computation is in regards to \emph{mean convergence time}. 
Let $T_{z}(t) = \inf \{t \geq 0 : X_t = \{0, n\} | X_0 = z\}$, and let $\tau_z = \mathbb{E}[T_{z}(t)]$. $\tau_z$ is the mean time to reach either absorbing state, starting from state $z$, which corresponds to the mean convergence time. The next theorem characterizes this distribution.

\begin{theorem}
\label{theorem:mean-convergence-time}
Let $\tau_z$ be the expected time to convergence, starting from state $z > n/2$, to any of the two converging states in the network (all-red or all-blue). Then, 
\begin{equation}
\tau_z = \ddfrac{\sum_{d = 1}^{n-1}x(d)y(d)}{2\sum_{l = 1}^{n/2}\prod_{i=1}^{l-1}\mu_i^2\prod_{j=l}^{n-l}\mu_j}
\end{equation}
where $x(d)$ and $y(d)$ are
\begin{equation}
\begin{split}
x(d) &= \sum_{l = 1}^{\min(z, d)} \prod_{i=1}^{l-1} \mu_i \prod_{j = l}^{d-1} \lambda_j\\
y(d) &= \sum_{l = 1}^{n - d - \max(z-d, 0)} \prod_{i = d+1}^{n-l} \mu_i \prod_{j = n - l + 1}^{n - 1} \lambda_j
\end{split}
\end{equation}
\end{theorem}

\begin{proof}
Following the calculations from before, $-B^{-1}$ at row $z$ provides the number of traversals to each other state starting from $z$. Calculating their sum, we have our result. The above equation is the full expression of the matrix row sum. 
\end{proof}

Theorem~\ref{theorem:mean-convergence-time} leads to the next lemma that captures property P2, under the assumption that at the beginning of the protocol, one proposal has at least $\alpha$ support in the network. 
\begin{lemma}
Slush reaches an absorbing state in finite time almost surely.
\label{lemma:finitetermination}
\end{lemma}

\begin{proof}
Starting from any non-absorbing, transient state, there is a non-zero probability of being absorbed. Additionally, since termination is finite and everywhere differentiable, Theorem~\ref{theorem:mean-convergence-time} also implies that the probability of termination of any network configuration where a proposal has $\geq \alpha$ support in bounded time $t_{max}$ is strictly positive. 
\end{proof}

\subsection{Snowflake}
\label{subsection:appendix_snowflake}
In Snowflake, the sampled set of nodes includes Byzantine nodes.
We introduce the decision function $\mathcal{D}(*)$, which is constructed by having each node also keep track of the total number of consecutive times it has sampled a majority of the same color ($\beta$). 
Finally, we introduce a function called $\mathcal{A}(\mathcal{S}_t)$, the adversarial strategy, that takes as parameters the entire configuration of the network at time $t$, as well as the next set of nodes chosen by the scheduler to execute, and as a side-effect, modifies the set of nodes $\mathcal{B}$ to some arbitrary configuration of colors.

In order for our prior framework to apply to Snowflake, we must deal with a key subtlety. 
Unlike in Slush, where it is clear that once the network has reached one of the converging states and therefore may not revert back, this no longer applies to Snowflake, since any adversary $f \geq \alpha$ has strictly positive probability of reverting the system, albeit this probability may be infinitesimally small. 
The CTMC is flexible enough to deal with a system where there is only one absorbing state, but the long-term behavior of the system is no longer meaningful since, after an infinite amount of time, the system is guaranteed to revert, violating safety. 
We could trivially bound the amount of time, and show safety using this bounded time assumption by simply characterizing the distribution of $e^{tQ}$, where $Q$ is the generator. 
However, we can make the following observation: if the probability of going from state $c$ (all-blue) to $c-1$ is exponentially small, then it will take the attacker exponential time (in expectation; note, this is a lower bound, and in reality it will take much longer) to succeed in reverting the system. 
Hence, we can assume that once all correct nodes are the same color, the attack from the adversary will terminate since it is impractical to continue an attack. 
In fact, under reasonably bounded timeframes, the variational distance between the exact approach and the approximation is very small. 
We leave details to the accompanying paper, but we briefly discuss how analysis proceeds for Snowflake. 

As stated in Section~\ref{sec:analysis}, the way to analyze the adversary using the same construction as in Slush is to condition reversibility on the first node $u$ deciding on blue, which can happen at any state (as specified by $\mathcal{D}(*)$). At that point, the adversarial strategy collapses to a single function, which is to continually vote for red. The probabilities of reversibility, for all states $\{1, \dots, c-1\}$ must encode the probability that additional blue nodes commit, and the single function of the adversary. The birth and death rates are transformed as follows:
\[
    \begin{cases}
        \mu_i = &i(1 - \mathbb{I}[\mathcal{D}(*, i, \mathbb{B})])\ \mathcal{H}(\mathcal{N}, c-i + f, k, \alpha)\\
        \lambda_i = &(c-i)(1 - \mathbb{I}[\mathcal{D}(*, c-i, \mathbb{R})])\ \mathcal{H}(\mathcal{N}, i, k, \alpha)\\
    \end{cases}
\]
From here on, the analysis is the same as in Slush. Under various $k$ and $\beta$, we can find the minimal $\alpha$ that provides the system strong irreversibility properties. 

The next lemma captures P3, and the proof follows from central limit theorem. 
\begin{lemma}
If $f < \Oh{\sqrt{n}}$, and $\alpha = \floor{k/2} + 1$, then Snowflake terminates in $\Oh{\log n}$ rounds with high probability. 
\label{lemma:centrallimit}
\end{lemma}
\begin{proof}
The results follows from central limit theorem, wherein for $\alpha = \floor{k/2} + 1$, the expected bias in the network after sampling will be $\Oh{\sqrt{n}}$. An adversary smaller than this bias will be unable to keep the network in a fully-bivalent state for more than a constant number of rounds. The logarithmic factor remains from the mixing time lower bound. 
\end{proof}

\subsection{Snowball}
We make the following observation: if the confidences between red and blue are equal, then the adversary has the same identical leverage in the irreversibility of the system as in Snowflake, regardless of network configuration. In fact, Snowflake can be viewed as Snowball but where drifts in confidences never exceed one. The same analysis applies to Snowball as in Snowflake, with the additional requirement of bounding the long-term behavior of the confidences in the network. To that end, analysis follows using martingale concentration inequalities, in particular the one introduced in Equation~\ref{eq:submartingale}. Snowball can be viewed as a two-urn system, where each urn is a sub-martingale. The guarantees that can be extracted hereon are that the confidences of the majority committed value (in our frame of reference is always blue), grow always more than those of the minority value, with high probability, drifting away as $t \rightarrow t_{max}$. 
\subsection{Safe Early Commitment}
As we reasoned previously, each conflict set in Avalanche can be viewed as an instance of Snowball, where each progeny instance iteratively votes for the entire path of the ancestry.
This feature provides various benefits; however, it also can lead to some virtuous transactions that depend on a rogue transaction to suffer the fate of the latter.
In particular, rogue transactions can interject in-between virtuous transactions and reduce the ability of the virtuous transactions to ever reach the required $\textsc{isAccepted}$ predicate.
As a thought experiment, suppose that a transaction $T_i$ names a set of parent transactions that are all decided, as per local view.
If $T_i$ is sampled over a large enough set of successful queries without discovering any conflicts, then, since by assumption the entire ancestry of $T_i$ is decided, it must be the case (probabilistically) that we have achieved irreversibility.

To then statistically measure the assuredness that $T_i$ has been accepted by a large percentage of correct nodes without any conflicts, we make use of a one-way birth process, where a birth occurs when a new correct node discovers the conflict of $T_i$. Necessarily, deaths cannot exist in this model, because a conflicting transaction cannot be unseen once a correct node discovers it. 
Our births are as follows:
\begin{equation}
    \lambda_i = \frac{c - i}{c} \left(1 - \frac{{n - i \choose k}}{{n \choose k}}\right)
\end{equation}
Solving for the expected time to reach the final birth state provides a lower bound to the $\beta_1$ parameter in the $\textsc{isAccepted}$ fast-decision branch. The table below shows an example of the analysis for $n = 2000$, $\alpha = 0.8k$, and various $k$, where $\varepsilon \ll 10^{-9}$, and where $\beta$ is the minimum required value before deciding.
\begin{table}[h!]
    \small
	\centering
	\begin{tabular}{llllll}
		$k$   & 10 & 20 & 30 & 40 \\ \hline
		$\beta$ & 10.87625 & 10.50125 & 10.37625 & 10.25125
	\end{tabular}
	\label{table:fast-path-beta}
\end{table}
\noindent Overall, a very small number of iterations are sufficient for the safe early commitment predicate. This supports the choice of $\beta$ in our evaluation.

\subsection{Initialization Heuristic}
\label{sec:sync-heuristic}
To improve liveness, we can use strong synchrony assumptions. The heuristic works as follows. Every node operates in two phases: in the first phase, it gossips and collects proposals for $\Oh{\log{n}}$ rounds, where each round lasts for the maximum message delay, which ensures the proposal from a correct node will propagate to almost all other correct nodes; in the second phase, each node stops collecting proposals, and gossips the existing proposals for an additional $\Oh{\log{n}}$ rounds so that every correct node ends up with approximately the same set of proposals. Finally, each node samples for the proposals it knows of locally, checking for those that have an $\alpha$ majority, ordered deterministically, such as by hash values. It then selects the first value by the order as its initial state when it starts the actual consensus protocol.
In a cryptocurrency setting, the deterministic ordering function would incorporate fees paid out for every new proposal, which means that the adversary is financially limited in its ability to launch a fairness attack against the initialization.

\subsection{Churn and View Discrepancy}\label{sec:full-analysis-churn}

Realistic systems need to accommodate the departure and arrival of nodes.
Up to now, we simplified our analysis by assuming a precise knowledge of network membership, i.e. $\mathcal{L}(u) = \mathcal{N}$.
We now demonstrate that correct nodes can admit a well-characterized amount of churn, by showing how to pick parameters such that Avalanche nodes can differ in their view of the network and still safely make decisions.

{\color{black} 

To characterize churn we use a generalized set intersection construction that allows us to make arguments about worst-case network view splits. Before formalizing it, we provide the intuition: suppose we split the network into two entirely independent, but fully connected, subsets. Clearly, the Byzantine adversary wins with probability one since it can send two conflicting transactions to the two independent networks respectively, and they would finalize the transactions immediately. This represents the worst case view split. 
\begin{figure}
\centering
\includegraphics[width=0.8\linewidth]{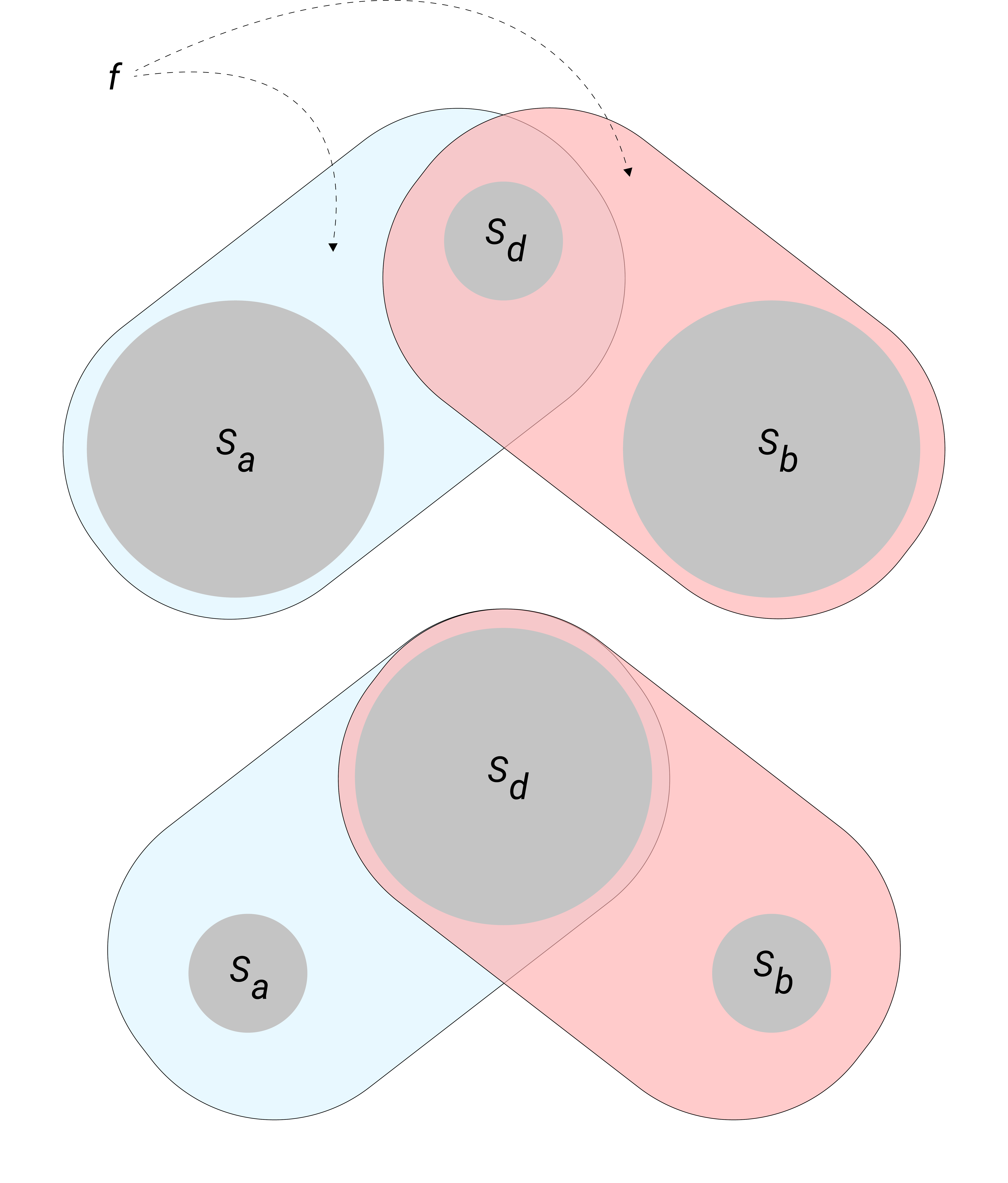}
\caption{Changes in network view based on $S_d$'s size. All prior proofs up to now represent the variant where $S_a = S_b = \emptyset$. With the new construction, the probability of safety violation is simply a direct application of the prior sets of proofs under the new subsets.}
\label{fig:network_view}
\end{figure}
We can generalize this to arbitrary network splits, which can even be applied recursively in each subset. The proofs of safety then are a matter of characterizing the probability of red and blue committing into the two (possibly independent) subsets.

Suppose we divide the set of correct nodes into three subsets, $S_a$, $S_d$, $S_b$. We overload $\mathcal{L}(S_{*})$ to represent the views of any node within the input set. The view of all nodes in $S_a$ is $\mathcal{L}(S_a) = S_a \cup S_d \cup \mathcal{B}$, the view of all nodes in $S_b$ is $\mathcal{L}(S_b) = S_b \cup S_d \cup \mathcal{B}$, and the view of $S_d$ is $\mathcal{N}$. We assume the worst case, which means that adversarial nodes are common to all subsets. When $S_d = \emptyset$, this represents a division of the network into two equal subsets where $|S_a| = |S_b| = n/2$~\footnote{$|S_a|$ and $|S_b|$ do not have to be equal, we assume so as a demonstration.}. If $S_d$ is all correct nodes, then $|S_a| = |S_b| = n$. This construction is visually demonstrated in Figure~\ref{fig:network_view}. 

\begin{lemma}
Let $\tau \in \mathbb{Z}^{+}$. Let $|S_d| = n - \tau$, and thus $|S_{\{a, b\}}| = \tau/2$. There exists some maximal size of $\tau$ such that probability of any two nodes $u, v \in S_a,\ S_b,\ S_d$ finalizing equivocating transactions is less than $\epsilon$. 
\end{lemma}
\begin{proof}
We assume that the adversary has full control of the network view splits, meaning that they choose in full how to create $S_a,\ S_b,\ S_d$. To prove safety, we simply reuse the same exact construction as in Subsection~\ref{subsection:appendix_snowflake}, but we replace the original set $\mathcal{N}$ with a new set of interest, namely $S_a \cup S_d \cup \mathcal{B}$ (i.e., we exclude $S_b$)~\footnote{Sets are symmetric in this example.}. To thus find the maximal $\tau$, we simply replace $u_i$ and $\lambda_i$ with 
\begin{equation}
    \begin{cases}
        \bar \mu_i = i\ \mathcal{H}(S_a \cup S_d \cup \mathcal{B}, c-(\tau/2)-i, k, \alpha), & \text{for } i \rightarrow i - 1 \\
        \lambda_i = (c-i)\ \mathcal{H}(S_a \cup S_d \cup \mathcal{B}, i, k, \alpha), & \text{for } j \rightarrow i + 1 \\
    \end{cases}
\end{equation}
where
\begin{equation}
\begin{split}
P(\mathcal{H}&(S_a \cup S_d \cup \mathcal{B}, x, k) \geq \alpha)\\
&= \left.\sum_{j = \alpha}^{k} {x \choose j}{c-(\tau/2)-x \choose k-j}\middle/{c - (\tau/2) \choose k}\right.\\
\end{split}
\end{equation}
and apply the same construction as in Subsection~\ref{subsection:appendix_snowflake}. As $\tau$ increases, it is clear that $S_a$ (and conversely, $S_b$), become more independent and less reliant on the values proposed by members in $S_d$, thus incrementing the ability of the adversaries. 
\end{proof}

}

\end{appendices}
\end{document}